\RequirePackage{fix-cm}
\documentclass[smallextended]{svjour3}       % onecolumn (second format)
\smartqed  % flush right qed marks, e.g. at end of proof
\usepackage{graphicx}
\usepackage{algorithmic}
\usepackage{algorithm}
\usepackage{amsmath}
\begin{document}

\title{A Scalable Heuristic for Viral Marketing Under the Tipping Model}

%\titlerunning{Short form of title}        % if too long for running head

\author{Paulo Shakarian \and
        Sean Eyre \and
        Damon Paulo
}

%\authorrunning{Short form of author list} % if too long for running head

\institute{P. Shakarian \at
              Network Science Center and\\
              Dept. Electrical Engineering and Computer Science \\
              U.S. Military Academy\\
              West Point, NY 10996\\
              Tel.: 845-938-5576\\
              \email{paulo@shakarian.net}           %  \\
           \and
           S. Eyre \at
              Network Science Center and\\
              Dept. Electrical Engineering and Computer Science \\
              U.S. Military Academy\\
              West Point, NY 10996\\
              Tel.: 845-938-5576\\
              \email{sean.k.eyre@gmail.com}           %  \\
             \and
              Damon Paulo \\
              Network Science Center and\\
              Dept. Electrical Engineering and Computer Science \\
              U.S. Military Academy\\
              West Point, NY 10996\\
              Tel.: 845-938-5576\\
              \email{damon.paulo@usma.edu}           %  \\
}

%\date{Received: date / Accepted: date}
% The correct dates will be entered by the editor

\maketitle

\begin{abstract}
In a ``tipping'' model, each node in a social network, representing an individual, adopts a property or behavior if a certain number of his incoming neighbors currently exhibit the same.  In viral marketing, a key problem is to select an initial "seed" set from the network such that the entire network adopts any behavior given to the seed.  Here we introduce a method for quickly finding seed sets that scales to very large networks.  Our approach finds a set of nodes that guarantees spreading to the entire network under the tipping model.  After experimentally evaluating 31 real-world networks, we found that our approach often finds seed sets that are several orders of magnitude smaller than the population size and outperform nodal centrality measures in most cases.  In addition, our approach scales well - on a Friendster social network consisting of 5.6 million nodes and 28 million edges we found a seed set in under 3.6 hours.  Our experiments also indicate that our algorithm provides small seed sets even if high-degree nodes are removed.  Lastly, we find that highly clustered local neighborhoods, together with dense network-wide community structures, suppress a trend's ability to spread under the tipping model.
\keywords{social networks \and viral marketing \and tipping model}
% \PACS{PACS code1 \and PACS code2 \and more}
% \subclass{MSC code1 \and MSC code2 \and more}
\end{abstract}
\section{Introduction}
A much studied model in network science, tipping\cite{Gran78,Schelling78,jy05} (a.k.a. deterministic linear threshold\cite{kleinberg}) is often associated with ``seed'' or ``target'' set selection,~\cite{chen09siam} (a.k.a. the maximum influence problem).  In this problem, we have a social network in the form of a directed graph and thresholds for each individual.  Based on this data, the desired output is the smallest possible set of individuals (seed set) such that, if initially activated, the entire population will become activated (adopting the new property).  This problem is NP-Complete~\cite{kleinberg,Dreyer09} so approximation algorithms must be used.  Though some such algorithms have been proposed,~\cite{leskovec07,chen09siam,benzwi09,chen10} none seem to scale to very large data sets.  Here, inspired by shell decomposition,~\cite{ShaiCarmi07032007,InfluentialSpreaders_2010,baxter11} we present a method guaranteed to find a set of nodes that causes the entire population to activate - but is not necessarily of minimal size.  We then evaluate the algorithm on $31$ large, real-world, social networks and show that it often finds very small seed sets (often several orders of magnitude smaller than the population size).  We also show that the size of a seed set is related to Louvain modularity and average clustering coefficient.  Therefore, we find that dense community structure combined with tight-knit local neighborhoods inhibit the spreading of activation under the tipping model.  We also found that our algorithm outperforms the classic centrality measures and is robust against the removal of high-degree nodes.

The rest of the paper is organized as follows.  In Section~\ref{prelim-sec}, we provide formal definitions of the tipping model.  This is followed by the presentation of our new algorithm in Section~\ref{alg-sec}.  We then describe our experimental results in Section~\ref{res-sec}.  Finally, we provide an overview of related work in Section~\ref{rw-sec}.

\section{Technical Preliminaries}
\label{prelim-sec}
Throughout this paper we assume the existence of a \textit{social network,} $G=(V,E)$, where $V$ is a set of vertices and $E$ is a set of directed edges.  We will use the notation $n$ and $m$ for the cardinality of $V$ and $E$ respectively.  For a given node $v_i \in V$, the set of incoming neighbors is $\eta^{in}_i$, and the set of outgoing neighbors is $\eta^{out}_i$.  The cardinalities of these sets (and hence the in- and out-degrees of node $v_i$) are $d^{in}_i, d^{out}_i$ respectively.  We now define a threshold function that for each node returns the fraction of incoming neighbors that must be activated for it to become activate as well.

\begin{definition}[Threshold Function]
We define the \textbf{threshold function} as mapping from V to $ (0,1] $.  Formally: $ \theta: V \rightarrow (0,1] $.
\end{definition}

For the number of neighbors that must be active, we will use the shorthand $k_i$.  Hence, for each $v_i$, $k_i = \lceil \theta (v_i) \cdot d^{in}_i \rceil$.  We now define an \textit{activation function} that, given an initial set of active nodes, returns a set of active nodes after one time step.

\begin{definition}[Activation Function]
Given a threshold function, $ \theta $, an \textbf{activation function} $ A_{\theta} $ maps subsets of V to subsets of V, 
where for some $ V' \subseteq V $,
\begin{equation}
A_{\theta}(V') = V' \cup \{ v_i \in V\ s.t.\ |\eta^{in}_i \cap V'| \geq k_i \} 
\end{equation}
\end{definition}

We now define multiple applications of the activation function.

\begin{definition}[Multiple Applications of the Activation Function]
Given a natural number $ i > 0$, set $V' \subseteq V $, and threshold function, $ \theta $, we define the multiple applications of the activation function, ${A^i_{\theta}}(V')$, as follows:
\begin{equation}
 A^i_\theta(V') = \begin{cases} A_\theta(V' ) & \text{if $i=1$}  \\ A_\theta(A^{i-1}_\theta(V' )) & \text{otherwise} \end{cases} 
\end{equation}
\end{definition}

Clearly, when $ A^i_\theta(V')= A^{i-1}_\theta(V')$  the process has converged.  Further, this always converges in no more than $n$ steps (as, prior to converging, a process must, in each step, activate at least one new node).  Based on this idea, we define the function $\Gamma$ which returns the set of all nodes activated upon the convergence of the activation function.

\begin{definition}[$\Gamma$ Function]
Let j be the least value such that $ A^j_{\theta}(V') = A^{j-1}_{\theta}(V') $.  We define the function $\Gamma_\theta : 2^V \rightarrow 2^V$ as follows.
\begin{equation}
\mathbf{ \Gamma_\theta } (V') = A^j_{ \theta }(V')
\end{equation}
\end{definition}

We now have all the pieces to introduce our problem - finding the minimal number of nodes that are initially active to ensure that the entire set $V$ becomes active.

\begin{definition}[The MIN-SEED Problem]
The MIN-SEED Problem is defined as follows: given a threshold function, $ \theta $, return $ V' \subseteq V\ s.t.\
\Gamma_\theta (V') = V $, and there does not exist $ V'' \subseteq V $ where $ |V''| < |V'| $ and 
$ \Gamma_\theta (V'') = V $.
\end{definition}

The following theorem is from the literature~\cite{kleinberg,Dreyer09} and tells us that the MIN-SEED problem is NP-complete.

\begin{theorem}[Complexity of MIN-SEED~\cite{kleinberg,Dreyer09}]
MIN-SEED in NP-Complete.
\end{theorem}

\section{Algorithms}
\label{alg-sec}

In this section, we introduce an integer program that solved the MIN-SEED problem exactly and our new decomposition-based heuristic.

\subsection{Exact Approach}

Below we present SEED-IP, an integer program that if solved exactly, guarantees an exact solution to MIN-SEED (see Proposition~\ref{ip-corr}).  Though, in general, solving an integer program is also NP-hard, suggesting that an exact solution will likely take exponential time, good approximation techniques such as branch-and-bound exist and mature tools such as QSopt and CPLEX can readily take and approximate solutions to integer programs.

\begin{definition}[SEED-IP]
\begin{eqnarray}
\label{objFcn}& \min \sum_i x_{i,1},\,\,\,\,\,\,\,\,\,\,	 \mathit{w.r.t.}\\
\label{varConst}\forall i,t \in \{1,\ldots, n\}, & x_{i,t} \in \{0, 1\}  \\
\label{endConst}\forall i, & x_{i,n} = 1 \\
\label{actConst}\forall i, \forall t >0,&\,\,\,\,\,\,\,\,\,\, x_{i,t} \leq x_{i,t-1}+ \frac{1}{d^{in}_i \theta(v_i)}\sum_{v_j \in \eta^{in}_i}x_{j,t-1}
\end{eqnarray}
\end{definition}

\begin{proposition}
\label{ip-corr}
If $V'$ is a solution to MIN-SEED, then setting $\forall v_i \in V', x_{i,1}=1$ and $\forall v_i \notin V', x_{i,1}=0$ is a solution to SEED-IP.\\
If the vector $[x_{i,t}]$ is a solution to SEED-IP, then $\{ v_i | x_{i,1}=1\}$ is a solution to MIN-SEED.
\end{proposition}
\begin{proof}
\noindent Claim 1: If $V'$ is a solution to MIN-SEED, then setting $\forall v_i \in V', x_{i,1}=1$ and $\forall v_i \notin V', x_{i,1}=0$ is a solution to SEED-IP.\\
Let $[x_{i,t}]$ be a vector for SEED-IP created as per claim 1.  Suppose, by way of contradiction (BWOC), there exists vector $[x'_{i,t}]$ s.t. $\sum_i x'_{i,1} < \sum_i x_{i,1}$.  However, consider the set of nodes $V'' = \{ v_i | x'_{i,1}=1\}$.  By Constraint~\ref{actConst} of SEED-IP, we know that, for $t >1$, that if $x'_{i,t} = 1$, we have $v_i \in A^t_\theta(V'')$.  Hence, by Constaint~\ref{endConst} $V''$ is a solution to MIN-SEED.  This means that $|V''| < |V'|$ as $\sum_i x'_{i,1} < \sum_i x_{i,1}$, which is a contradiction.\\
\noindent Claim 2: If the vector $[x_{i,t}]$ is a solution to SEED-IP, then $\{ v_i | x_{i,1}=1\}$ is a solution to MIN-SEED.\\
Suppose, BWOC, there exists set $V''$ that is a solution to MIN-SEED s.t. $|V''| < |\{ v_i | x_{i,1}=1\}|$.  Consider the vector $[x'_{i,t}]$ where $\forall i,$ $x'_{i,0}=1$ iff $v_i \in V''$.  By Constraint~\ref{actConst} of SEED-IP, we know that, for $t >1$, that if $v_i \in A^t_\theta(V'')$, we have $x'_{i,t} = 1$.  Hence, as $A^t_\theta(V'') = V$, know that $[x'_{i,t}]$  satisfies Constraint~\ref{endConst}.  Hence, as $|V''| < |\{ v_i | x_{i,1}=1\}|$, we know $\sum_i x'_{i,1} < \sum_i x_{i,1}$, which is a contradiction.\\
\noindent Proof of theorem: Follows directly form claims 1-2.
\end{proof}

However, despite the availability of approximate solvers, SEED-IP requires a quadratic number of variables and constraints (Proposition~\ref{numConst}), which likely will prevent this approach from scaling to very large datasets.  As a result, in the next section we introduce our heuristic approach.

\begin{proposition}
\label{numConst}
SEED-IP requires $n^2$ variables and $2n^2$ constraints.
\end{proposition}

\subsection{Heuristic}

To deal with the intractability of the MIN-SEED problem, we design an algorithm that finds a non-trivial subset of nodes that causes the entire graph to  activate, but we do not guarantee that the resulting set will be of minimal size.  The algorithm is based on the idea of shell decomposition often cited in physics literature~\cite{Seidman83,ShaiCarmi07032007,InfluentialSpreaders_2010,baxter11} but modified to ensure that the resulting set will lead to all nodes being activated.  The algorithm, \textsf{TIP\_DECOMP} is presented in this section.

\algsetup{indent=1em}
	\begin{algorithm}[h!]
		\caption{ \textsf{TIP\_DECOMP}}
		\begin{algorithmic}[1]

		\REQUIRE Threshold function, $ \theta $ and directed social network $G=(V,E)$
		\ENSURE $ V' $
		\medskip

		\STATE{ For each vertex $ v_i $, compute $ k_i $}.
		\STATE{ For each vertex $ v_i,\ dist_i = d_i^{in} - k_i $}.
		\STATE{ FLAG = TRUE}.
		\WHILE{ FLAG }
			\STATE {Let $ v_i $ be the element of $ v $ where $ dist_i $ is minimal}.
			\IF { $ dist_i = \infty $}
				\STATE{ FLAG = FALSE}.
			\ELSE
				\STATE{ Remove $ v_i $ from $G$ and for each $ v_j $ in $ \eta_i^{out} $, if $dist_j > 0$, set
				$  dist_j = dist_j-1 $.  Otherwise set $dist_j = \infty $}.
			\ENDIF
		\ENDWHILE
		\RETURN{ All nodes left in $ G $}.
	\end{algorithmic}
\end{algorithm}

Intuitively, the algorithm proceeds as follows (Figure 1).  Given network $G=(V,E)$ where each node $v_i$ has threshold $k_i = \lceil \theta (v_i) \cdot d^{in}_i \rceil$, at each iteration, pick the node for which $d^{in}_i - k_i$ is the least but positive (or $0$) and remove it.  Once there are no nodes for which $d^{in}_i - k_i$ is positive (or $0$), the algorithm outputs the remaining nodes in the network.  

\begin{figure}
    \begin{center}
        \includegraphics[width=1\linewidth]{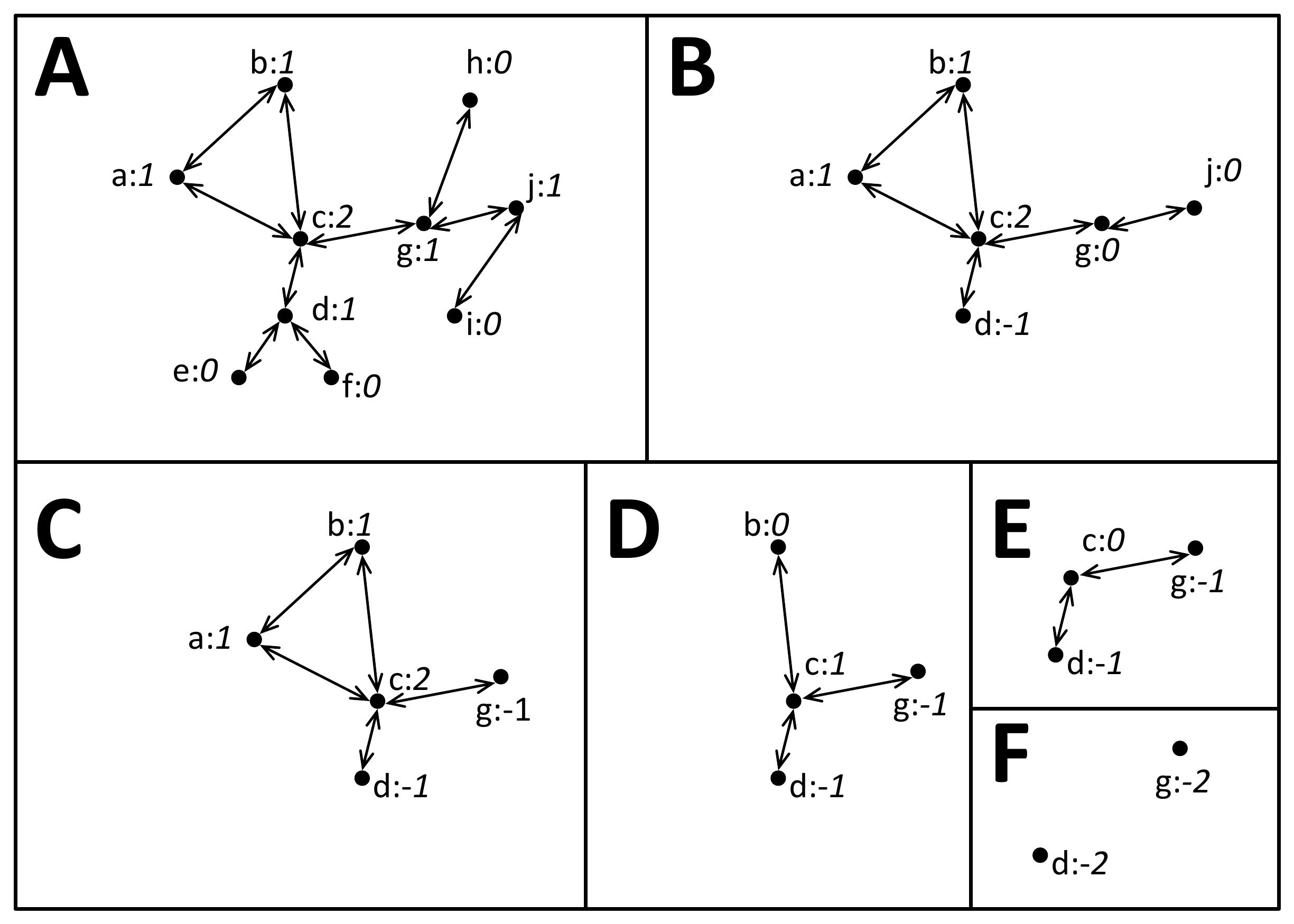}
    \end{center}
\caption{Example of our algorithm for a simple network depicted in box \textbf{A}.  We use a threshold value set to $50\%$ of the node degree.  Next to each node label (lower-case letter) is the value for $d^{in}_i - k_i$ (where $k_i = \lceil \frac{d^{in}_i}{2} \rceil$).  In the first four iterations, nodes e, f, h, and i are removed resulting in the network in box \textbf{B}.  This is followed by the removal of node j resulting in the network in box \textbf{C}.  In the next two iterations, nodes a and b are removed (boxes \textbf{D}-\textbf{E} respectively).  Finally, node c is removed (box \textbf{F}).  The nodes of the final network, consisting of d and g, have negetive values for $d_i-\theta_i$ and become the output of the algorithm.}
    %\label{main-fig}
\end{figure}

Now, we prove that the resulting set of nodes is guaranteed to cause all nodes in the graph to activate under the tipping model.  This proof follows from the fact that any node removed is activated by the remaining nodes in the network.  

\begin{theorem}
If all nodes in $ V'\ \subseteq\ V $ returned by \textsf{TIP\_DECOMP} are initially active, then every node in $ V $ will eventually be activated, too. 
\end{theorem}

\begin{proof}
Let $ w $ be the total number of nodes removed by \textsf{TIP\_DECOMP}, where $ v_1 $ is the last node removed and $ v_w $ is the first node removed.  We prove the theorem by induction on $ w $ as follows.  We use $P(w)$ to denote the inductive hypothesis which states that all nodes from $ v_1 $ to $ v_w $ are active. In the base case, $P(1)$ trivially holds as we are guaranteed that from set $V'$ there are at least $k_1$ edges to $v_1$ (or it would not be removed).  For the inductive step, assuming $P(w)$ is true, when $ v_{w+1} $ was removed from the graph $ dist_{w+1} \geq 0 $ which means that $ d_{w+1}^{in} \geq k_{w+1} $.  All nodes in $\eta^{in}_{w+1}$ at the time when $ v_{w+1} $ was removed are now active, so $ v_{w+1} $ will now be activated - which completes the proof.
\end{proof}

We also note that by using the appropriate data structure (we used a binomial heap in our implementation), for a network of $n$ nodes and $m$ edges, this algorithm can run in time $O(m \log n)$.

\begin{proposition}
\label{tcomp}
The complexity of \textsf{TIP\_DECOMP} is $ O(m \cdot log(n)) $.
\end{proposition}
\begin{proof}
If we use a binomial heap as described in \cite{cormen}, we can create a heap where we store each node and assign it a key value of $dist_i$ for each node $v_i$.  The creation of a heap takes constant time and inserting the $n$ vertices will take $O(n log(n))$ time.  We can also maintain a list data structure as well.  In the course of the while loop, all nodes will either be removed (as per the algorithm), decreased in key-value no more than $d^{in}_i$ or increased to infinity (which we can implement as being removed and added to the list).  Hence, the number of decrease key or removal operations is bounded by $n+\sum_i d^{in}_i$.  As $\sum_i d^{in}_i = m$ (where $m$ is the number of edges).  As $ O(m \cdot log(n)) $, the statement follows.
\end{proof}

\section{Results}
\label{res-sec}
In this section we describe the results of our experimental evaluation.  We describe the datasets we used for the experiments in Section~\ref{ds-sec}.  We evaluate the run-time of \textsf{TIP\_DECOMP} in Section~\ref{rt-sec}.  In Section~\ref{ss-sec}, we evaluate the size of the seed-set returned by the algorithm and we compare this to the seed size returned by known centrality measures in Section~\ref{cent-sec}.  The speed of the activiation process initiated with seed sets discovered by our algorithm is described in Section~\ref{speedSec}.  We then study how the removal of high-degree nodes and community structure affect the results of the algorithm in Sections~\ref{robust-sec} and \ref{comm-sec} respectively.

The algorithm \textsf{TIP\_DECOMP} was written using Python 2.6.6 in 200 lines of code that leveraged the NetworkX library available from\\ http://networkx.lanl.gov/.  The code used a binomial heap library written by Bj\"orn B. Brandenburg available from http://www.cs.unc.edu/$\sim$bbb/.  The experiments were run on a computer equipped with an Intel X5677 Xeon Processor operating at 3.46 GHz with a 12 MB Cache running Red Hat Enterprise Linux version 6.1 and equipped with 70 GB of physical memory.  All statistics presented in this section were calculated using R 2.13.1.

\subsection{Datasets}
\label{ds-sec}

In total, we examined $36$ networks: nine academic collaboration networks, three e-mail networks, and $24$ networks extracted from social-media sites.  The sites included included general-purpose social-media (similar to Facebook or MySpace) as well as special-purpose sites (i.e. focused on sharing of blogs, photos, or video).

All datasets used in this paper were obtained from one of four sources: the ASU Social Computing Data Repository,~\cite{Zafarani+Liu:2009} the Stanford Network Analysis Project,~\cite{snap} the University of Michigan,~\cite{umich} and Universitat Rovira i Virgili.\cite{uvi}  $31$ of the networks considered were symmetric -- i.e. if a directed edge from vertex $v$ to $v'$ exists, there is also an edge from vertex $v'$ to $v$.  Tables~\ref{fig3} (A-C) show some of the pertinent qualities of the symmetric networks.  The networks are categorized by the results for the MIN-SEED experiments (explained later in this section).  Additionally, we also looked at several non-symmetric (directed) networks and placed them in their own category.  In what follows, we provide their real-world context.

\subsubsection{Category A}
\begin{itemize}
\item{\textbf{BlogCatalog} is a social blog directory that allows users to share blogs with friends.~\cite{Zafarani+Liu:2009}  The first two samples of this site, BlogCatalog1 and 2, were taken in Jul. 2009 and June 2010 respectively.  The third sample, BlogCatalog3 was uploaded to ASU's Social Computing Data Repository in Aug. 2010.}
\item{\textbf{Buzznet} is a social media network designed for sharing photographs, journals, and videos.~\cite{Zafarani+Liu:2009}  It was extracted in Nov. 2010.}
\item{\textbf{Douban } is a Chinese social medial website designed to provide user reviews and recommendations.~\cite{Zafarani+Liu:2009}  It was extracted in Dec. 2010.}
\item{\textbf{Flickr} is a social media website that allows users to share photographs.~\cite{Zafarani+Liu:2009}  It was uploaded to ASU's Social Computing Data Repository in Aug. 2010.}
\item{\textbf{Flixster} is a social media website that allows users to share reviews and other information about cinema.~\cite{Zafarani+Liu:2009}  It was extracted in Dec. 2010.}
\item{\textbf{FourSquare} is a location-based social media site.~\cite{Zafarani+Liu:2009}  It was extracted in Dec. 2010.}
\item{\textbf{Frienster} is a general-purpose social-networking site.~\cite{Zafarani+Liu:2009} It was extracted in Nov. 2010.}
\item{\textbf{Last.Fm} is a music-centered social media site.~\cite{Zafarani+Liu:2009} It was extracted in Dec. 2010.}
\item{\textbf{LiveJournal} is a site designed to allow users to share their blogs.~\cite{Zafarani+Liu:2009}  It was extracted in Jul. 2010.}
\item{\textbf{Livemocha} is touted as the ``world's largest language community.''~\cite{Zafarani+Liu:2009}  It was extracted in Dec. 2010.}
\item{\textbf{WikiTalk} is a network of individuals who set and received messages while editing WikiPedia pages.~\cite{snap}  It was extracted in Jan. 2008.}
\end{itemize}

\subsubsection{Category B}
\begin{itemize}
\item{\textbf{Delicious} is a social bookmarking site, designed to allow users to share web bookmarks with their friends.~\cite{Zafarani+Liu:2009} It was extracted in Dec. 2010.}
\item{\textbf{Digg} is a social news website that allows users to share stories with friends.~\cite{Zafarani+Liu:2009}  It was extracted in Dec. 2010.}
\item{\textbf{EU E-Mail} is an e-mail network extracted from a large European Union research institution.~\cite{snap}  It is based on e-mail traffic from Oct. 2003 to May 2005.}
\item{\textbf{Hyves} is a popular general-purpose Dutch social networking site.~\cite{Zafarani+Liu:2009}  It was extracted in Dec. 2010.}
\item{\textbf{Yelp} is a social networking site that allows users to share product reviews.~\cite{Zafarani+Liu:2009}  It was extracted in Nov. 2010.}
\end{itemize}

\subsubsection{Category C}
\begin{itemize}
\item{\textbf{CA-AstroPh} is a an academic collaboration network for Astro Physics from Jan. 1993 - Apr. 2003.~\cite{snap}}
\item{\textbf{CA-CondMat} is an academic collaboration network for Condense Matter Physics.  Samples from 1999 (CondMat99), 2003 (CondMat03), and 2005 (CondMat05) were obtained from the University of Michigan.~\cite{umich}  A second sample from 2003 (CondMat03a) was obtained from Stanford University.~\cite{snap}}
\item{\textbf{CA-GrQc} is a an academic collaboration network for General Relativity and Quantum Cosmology from Jan. 1993 - Apr. 2003.~\cite{snap}}
\item{\textbf{CA-HepPh} is a an academic collaboration network for High Energy Physics - Phenomenology from Jan. 1993 - Apr. 2003.~\cite{snap}}
\item{\textbf{CA-HepTh} is a an academic collaboration network for High Energy Physics - Theory from Jan. 1993 - Apr. 2003.~\cite{snap}}
\item{\textbf{CA-NetSci} is a an academic collaboration network for Network Science from May 2006.}
\item{\textbf{Enron E-Mail} is an e-mail network from the Enron corporation made public by the Federal Energy Regulatory Commission during its investigation.~\cite{snap}}
\item{\textbf{URV E-Mail} is an e-mail network based on communications of members of the University Rovira i Virgili (Tarragona).~\cite{uvi}  It was extracted in 2003.}
\item{\textbf{YouTube} is a video-sharing website that allows users to establish friendship links.~\cite{Zafarani+Liu:2009}  The first sample (YouTube1) was extracted in Dec. 2008.  The second sample (YouTube2) was uploaded to ASU's Social Computing Data Repository in Aug. 2010.}
\end{itemize}

\subsubsection{Non-Symmetric Networks}
\begin{itemize}
\item{\textbf{Epinions} is a consumer review website that allows members to establish directed trust relationships.~\cite{snap}}
\item{\textbf{WikiVote} is a sample of Wikipedia users voting beahavior (who votes for whom).~\cite{snap}}
\item{\textbf{Slashdot} formerly had a feature called ``Slashdot Zoo'' that allowed users to tag each other as friend or foe.  We looked at three samples based on friendship relationships: one sample from 2008 (Slashdot1) and two from 2009 (Slashdot2-Slashdot3).~\cite{snap}}
\end{itemize}

\begin{table}[ht]
     \begin{center}
        \includegraphics[width=.8\linewidth]{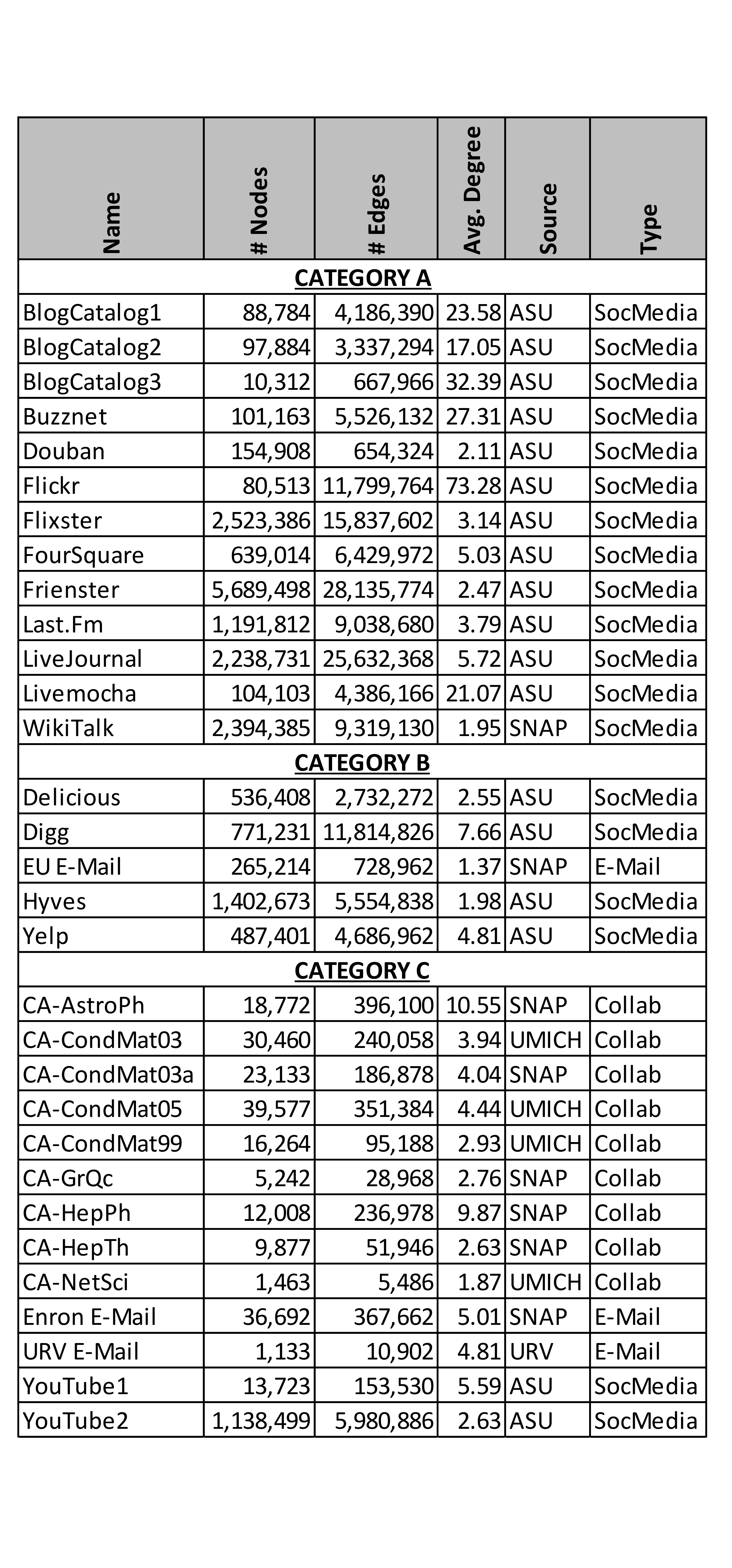}
    \end{center}
    \caption{Information on the networks in Categories A, B, and C.}
    \label{fig3}
\end{table}

\begin{table}[ht]
     \begin{center}
        \includegraphics[width=.8\linewidth]{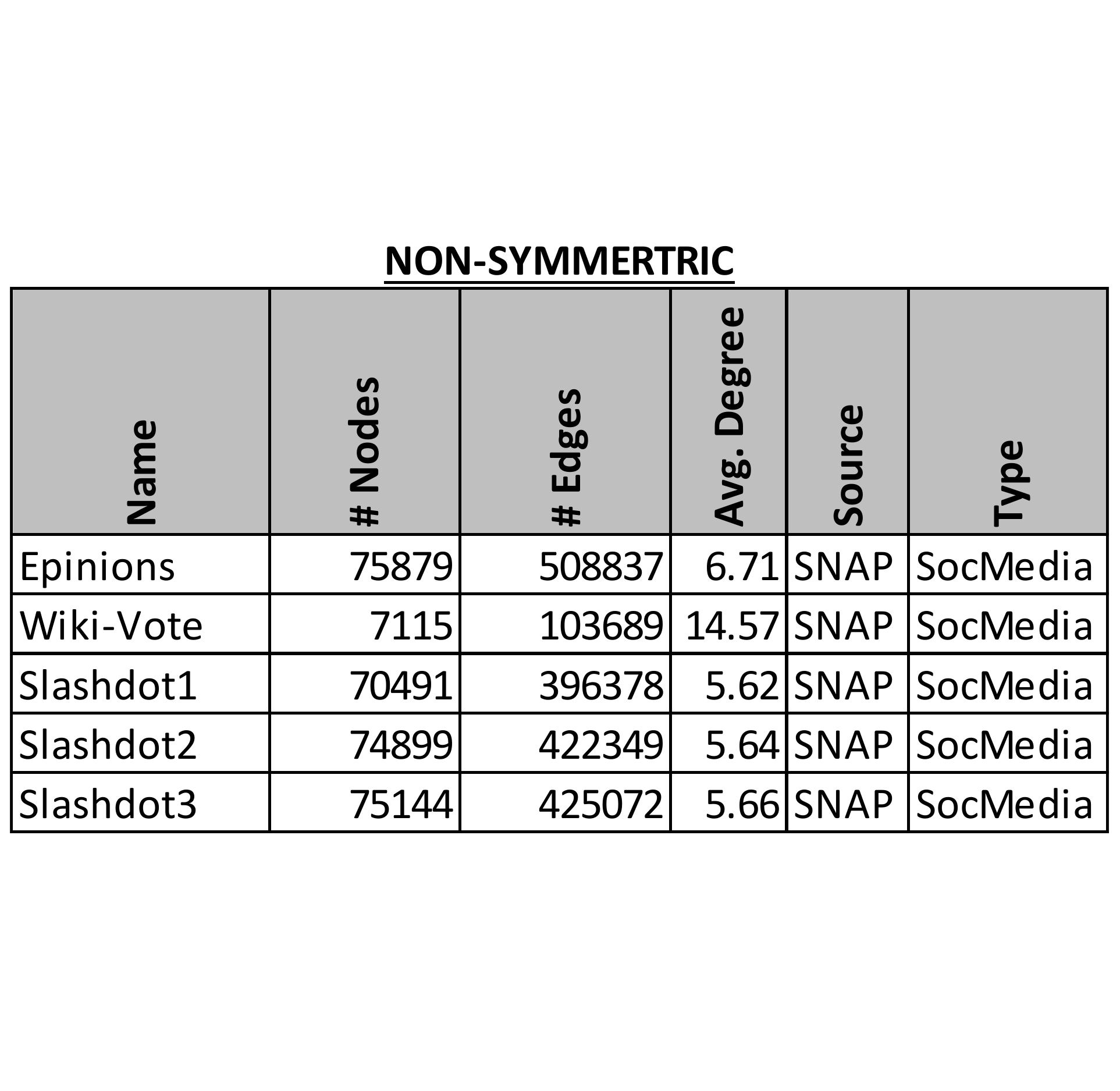}
    \end{center}
    \caption{Information on non-symmetric networks.}
    \label{dirInfo}
\end{table}

\subsubsection{Runtime}
\label{rt-sec}
First, we examined the runtime of the algorithm (see Figure~\ref{figRt} and Table~\ref{figRtTable}).  Our experiments aligned well with our time complexity result (Proposition~\ref{tcomp}).  For example, a network extracted from the Dutch social-media site Hyves consisting of $1.4$ million nodes and $5.5$ million directed edges was processed by our algorithm in at most $12.2$ minutes.  The often-cited LiveJournal dataset consisting of $2.2$ million nodes and $25.6$ million directed edges was processed in no more than $66$ minutes - a short time to approximate an NP-hard combinatorial problem on a large-sized input.

\begin{figure}[htbb]
    \begin{center}
        \includegraphics[width=.89\linewidth]{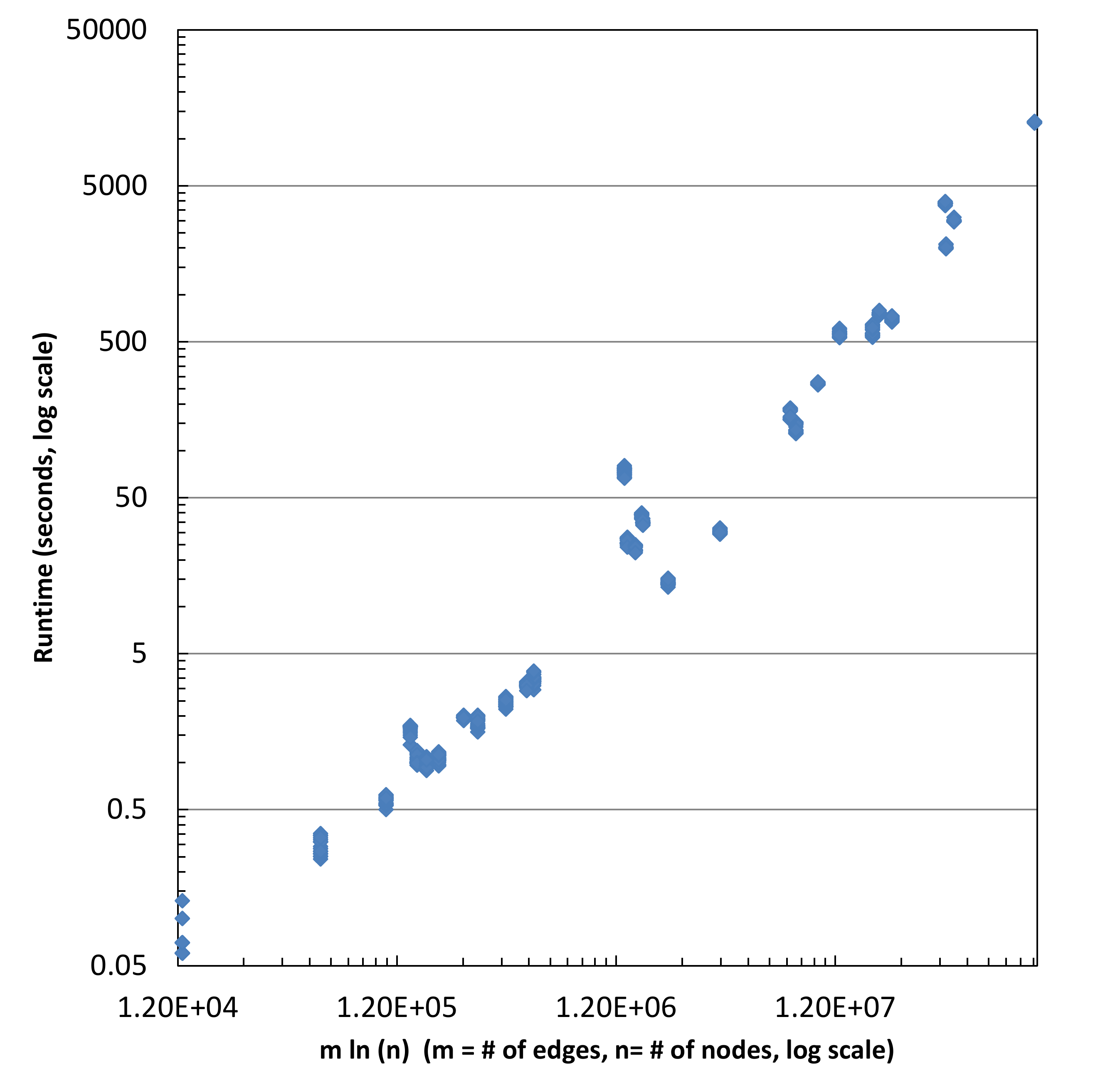}
    \end{center}
    \caption{$m \ln n$ vs. runtime in seconds (log scale, $m$ is number of edges, $n$ is number of nodes).  The relationship is linear with $R^2=0.9015$, $p=2.2 \cdot 10^{-16}$.}
    \label{figRt}
\end{figure}

\begin{table}[h]
     \begin{center}
        \includegraphics[width=.8\linewidth]{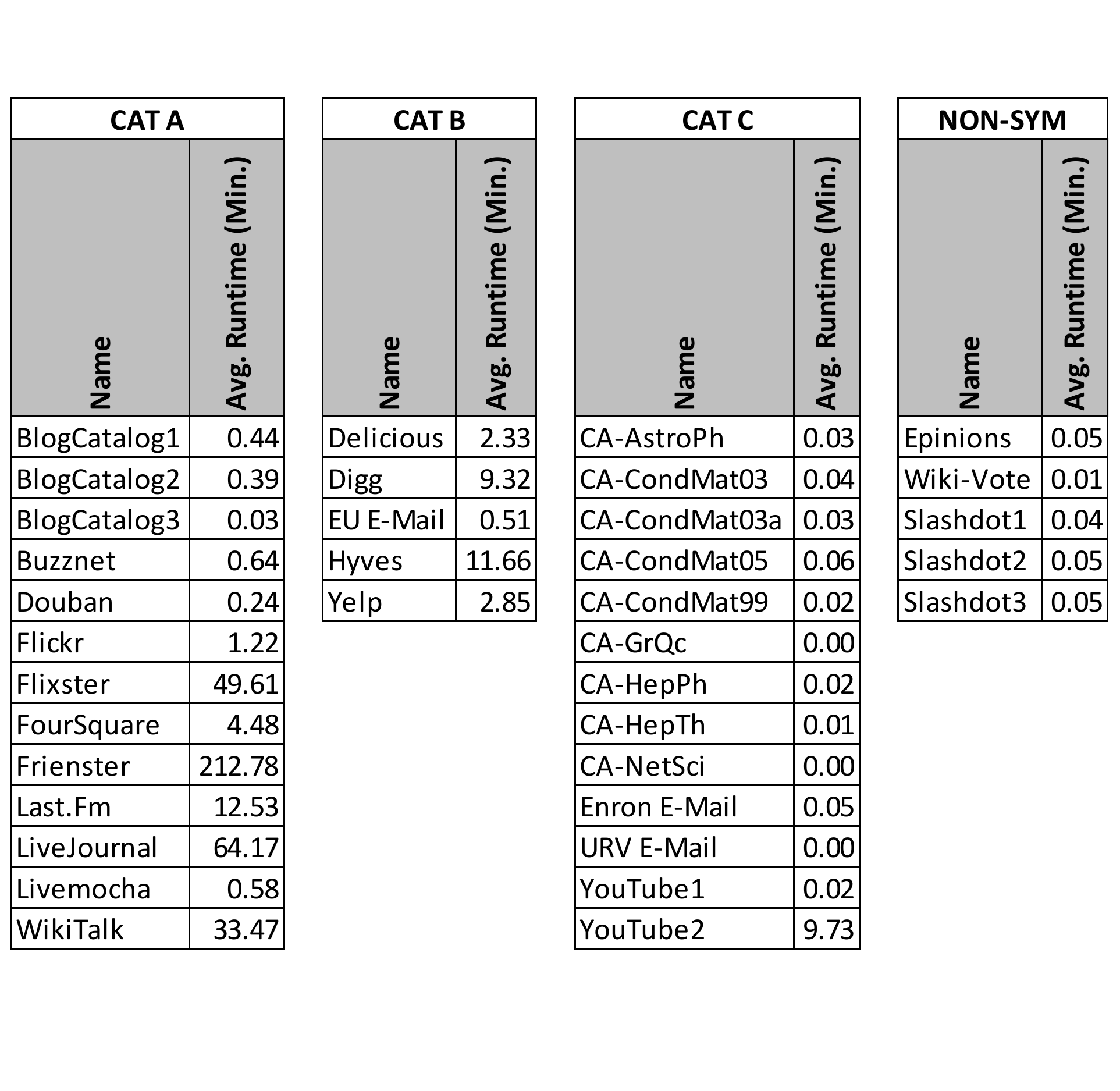}
    \end{center}
    \caption{Runtime data on the datasets used in the experiments.}
    \label{figRtTable}
\end{table}

\subsubsection{Seed Size}
\label{ss-sec}
For each network, we performed $10$ ``integer'' trials.  In these trials, we set $\theta(v_i)=\min(d^{in}_i,k)$ where $k$ was kept constant among all vertices for each trial and set at an integer in the interval $[1,10]$.  We evaluated the ability of a network to promote spreading under the tipping model based on the size of the set of nodes returned by our algorithm (as a percentage of total nodes).  For purposes of discussion, we have grouped our networks into three categories based on results (Figure~\ref{figFirst} and Table~\ref{figX}).  We have also included results for symmetric networks (Figure~\ref{figNonSym} and Table~\ref{dirResTable}).  In general, online social networks had the smallest seed sets - $13$ networks of this type had an average seed set size less than $2\%$ of the population (these networks were all in Category A).  We also noticed, that for most networks, there was a linear realtion between threshold value and seed size.

\begin{figure}
    \begin{center}
        \includegraphics[width=1\linewidth]{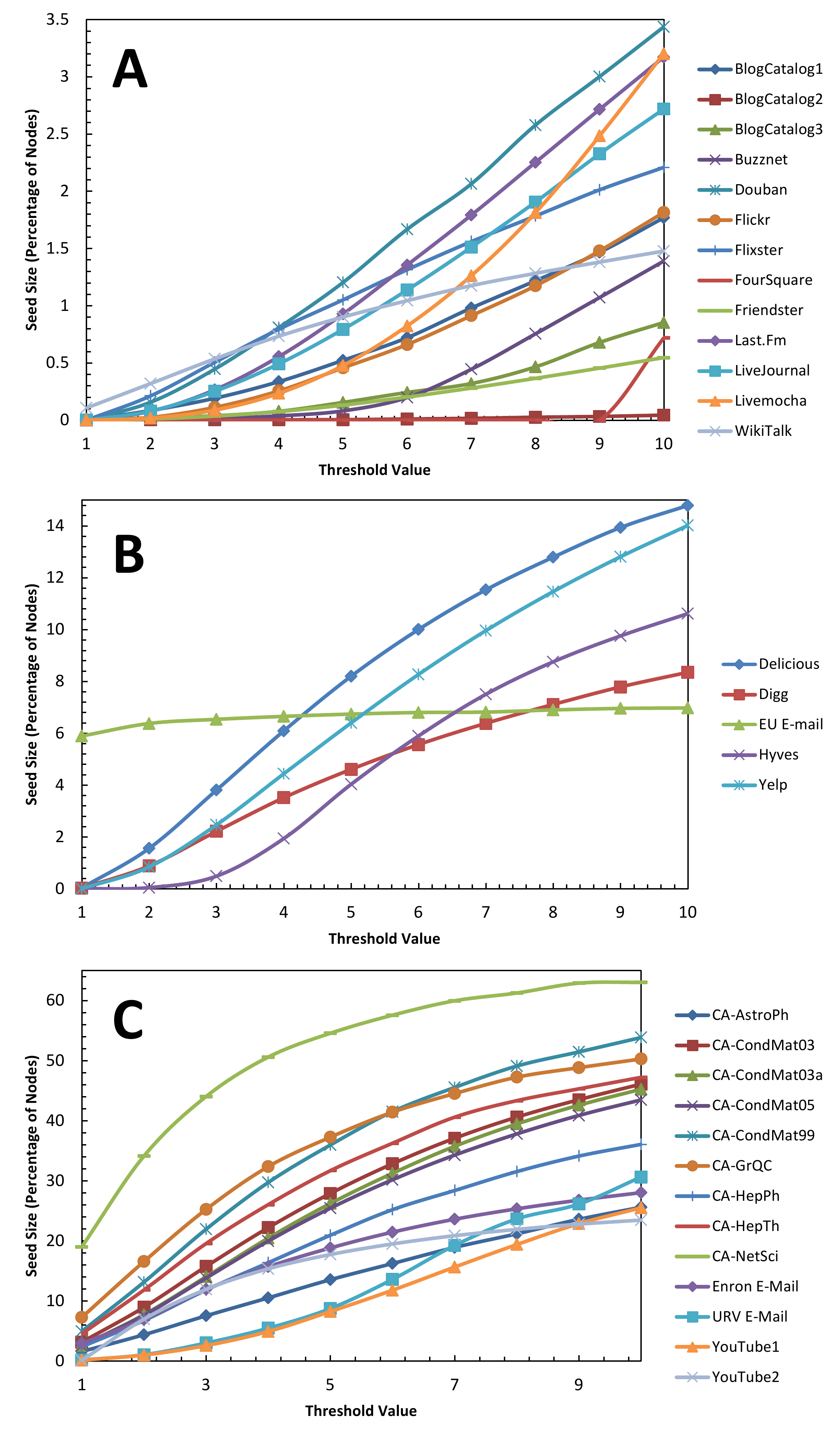}
    \end{center}
       \caption{Threshold value (assigned as an integer in the interval $[1,10]$) vs. size of initial seed set as returned by our algorithm in our three identified categories of networks (categories A-C are depicted in panels A-C respectively).  Average seed sizes were under $2\%$ for Categorty A, $2-10\%$ for Category B and over $10\%$ for Category C.  The relationship, in general, was linear for categories A and B and lograthimic for C.  CA-NetSci had the largest Louvain Modularity and clustering coefficient of all the networks.  This likely explains why that particular network seems to inhibit spreading.}
    \label{figFirst}
\end{figure}

\begin{figure}
    \begin{center}
        \includegraphics[width=1\linewidth]{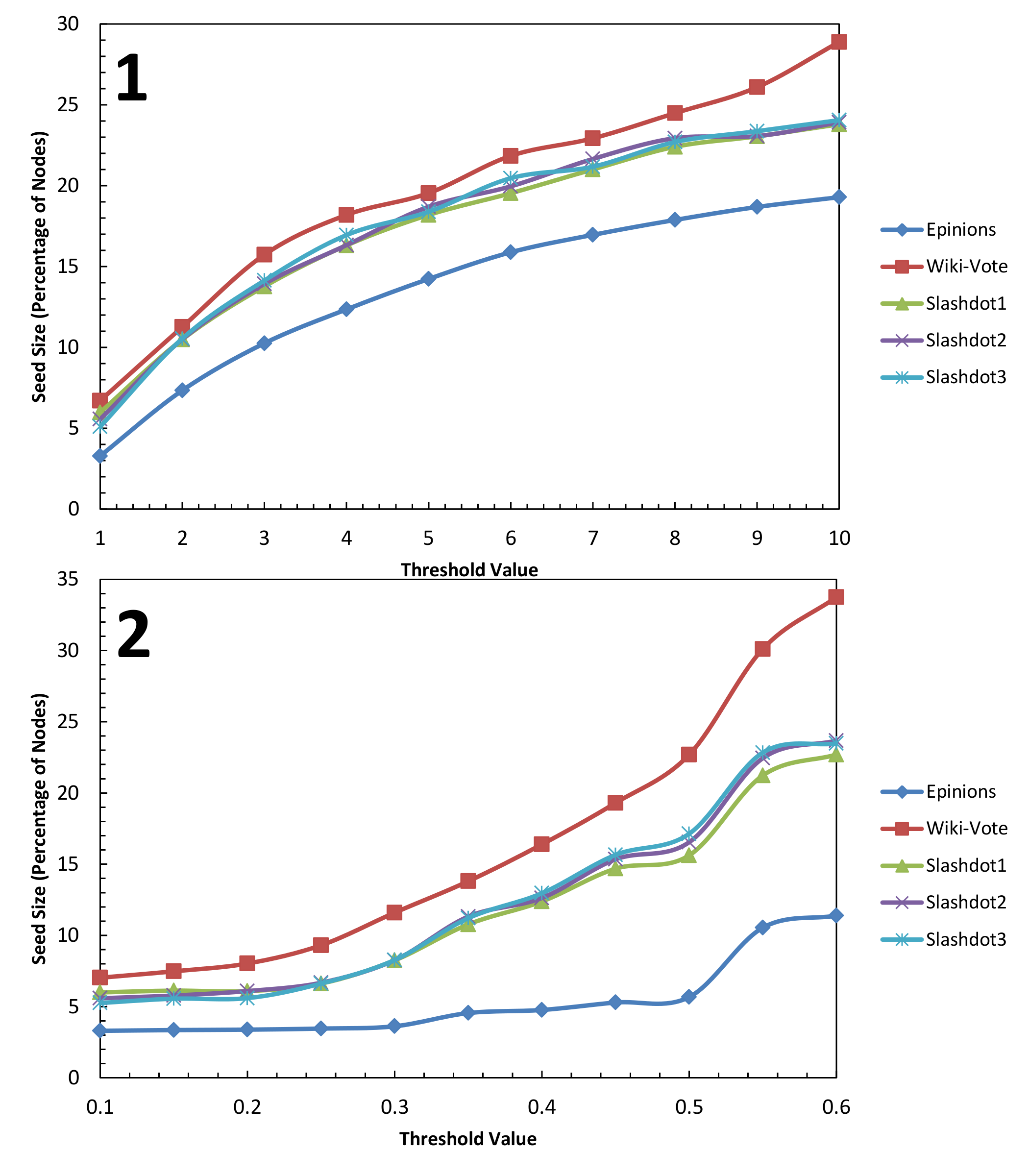}
    \end{center}
   \caption{Threshold value assigned as both an integer in the range $[1,10]$ (panel 1) as well as a fraction of node degree (panel 2) for the non-symmetric networks.}
    \label{figNonSym}
\end{figure}

Category A can be thought of as social networks highly susceptible to influence - as a very small fraction of initially activated individuals can lead to activation of the entire population.  All were extracted from social media websites.  For some of the lower threshold levels, the size of the set of seed nodes was particularly small.  For a threshold of three, $11$ of the Category A networks produced seeds smaller than $0.5\%$ of the total populationa.  For a threshold of four, nine networks met this criteria.

Networks in Category B are susceptible to influence with a relatively small set of initial nodes - but not to the extent of those in Category A.  They had an average initial seed size greater than $2\%$ but less than $10\%$. Members in this group included two general purpose social media networks, two specialty social media networks, and an e-mail network.  Non-symmetric networks generally perofrmed somewhat poorer than Category B networks (though in general, not as poorly as those in Category C).  The initial seed sizes for the non-symmmetric networks ranged from $3\%$ to $29\%$.

Category C consisted of networks that seemed to hamper diffusion in the tipping model, having an average initial seed size greater than $10\%$.  This category included all of the academic collaboration networks, two of the email networks, and two networks derived from friendship links on YouTube.

We also studied the effects on spreading when the threshold values were assigned as a specific fraction of each node's in-degree~\cite{jy05,wattsDodds07}, which results in heterogeneous $\theta_i$'s across the network.  We performed $12$ trials for each network.  Thresholds for each trial were based on the product of in-degree and a fraction in the interval $[0.05,0.60]$ (multiples of $0.05$).  The results (Figure~\ref{figFrac} and Table~\ref{figX}; for non-symmertic networks see Figure~\ref{figNonSym} and Table~\ref{dirResTable}) were analogous to our integer tests.  We also compared the averages over these trials with $M$ and $C$ and obtained similar results as with the other trials (Figure~\ref{main-fig} B).

\begin{figure}
    \begin{center}
        \includegraphics[width=1\linewidth]{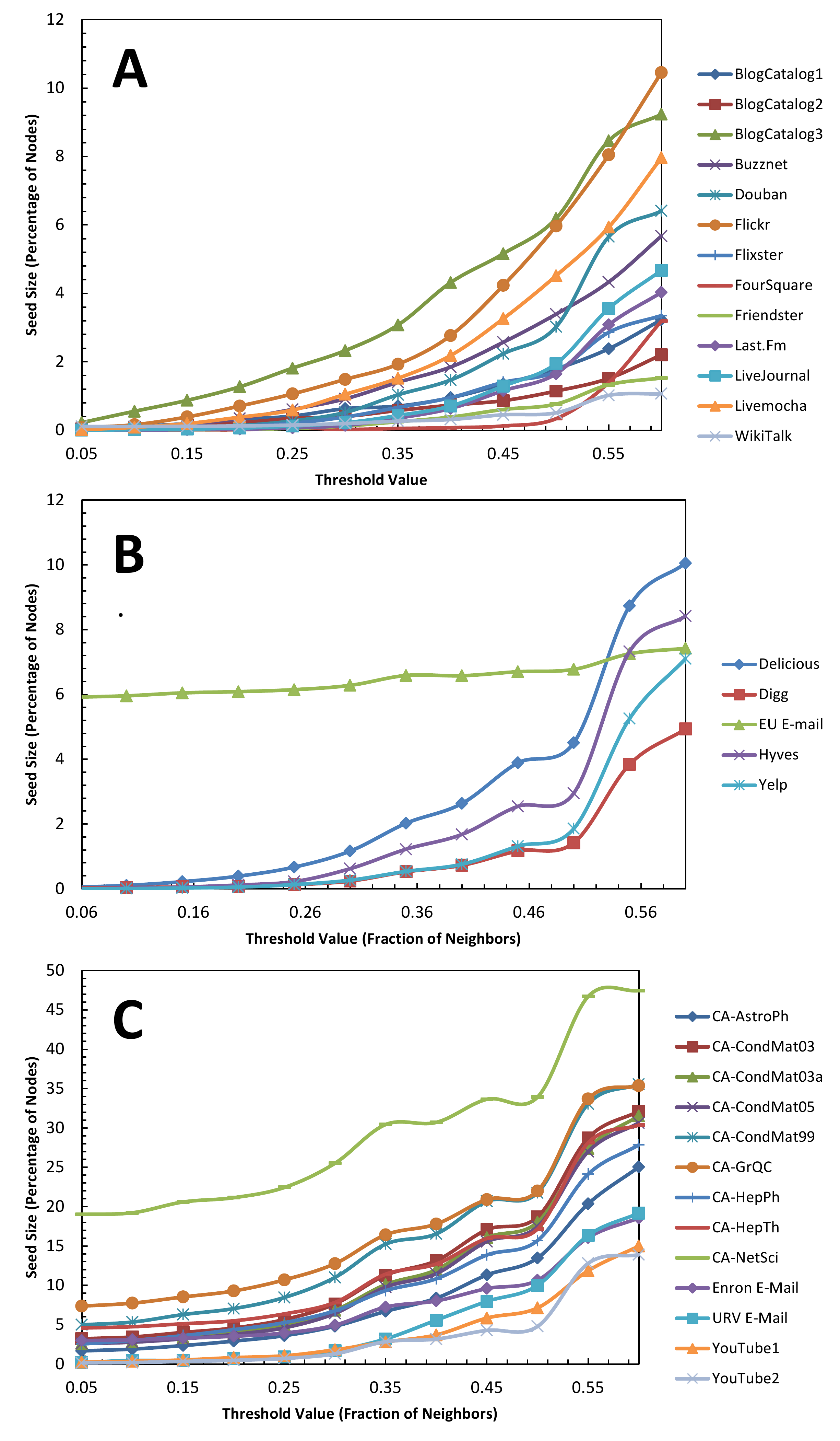}
    \end{center}
   \caption{Threshold value (assigned as a fraction of node in-degree as a multiple of $0.05$ in the interval $[0.05,0.60]$) vs. size of initial seed set as returned by our algorithm in our three identified categories of networks (categories A-C are depicted in panels A-C respectively, categories are the same as in Figure 1). Average seed sizes were under $5\%$ for Categorty A, $1-7\%$ for Category B and over $3\%$ for Category C.  In general, the relationship between threshold and initial seed size for networks in all categories was exponential.}
    \label{figFrac}
\end{figure}

\subsection{Comparison with Centrality Measures}
\label{cent-sec}

We compared our results with six popular centrality measures: degree, betweenness, closeness, shell number, eigenvector, and PageRank.  Here, we define degree centrality is simply the number of outgoing adjacent nodes.~\footnote{Note that in the symmetric networks we examined, our empirical results hold for the number of incoming adjacent edges as well as the total number of adjacent edges.}  The intuition behind high betweenness centrality nodes is that they function as ``bottlenecks'' as many paths in the network pass through them.  Hence, betweenness is a medial centrality measure.  Let $\sigma_{st}$ be the number of shortest paths between nodes $s$ and $t$ and $\sigma_{st}(v)$ be the number of shortest paths between $s$ and $t$ containing node $v$.  In \cite{freeman77}, betweenness centrality for node $v$ is defined as $\sum_{s\neq v \neq t}\frac{\sigma_{st}(v)}{\sigma_{st}}$.  In most implementations, including the ones used in this paper, the algorithm of \cite{brandes01} is used to calculate betweenness centrality.  Another common measure from the literature that we examined is closeness ~\cite{freeman79cent}.  Given node $i$, its closeness $C_c(i)$ is the inverse of the average shortest path length from node $i$ to all other nodes in the graph. Intuitively, closeness measures how ``close'' it is to all other nodes in a network.  Formally, if we define the shortest path between nodes $i$ to $j$ as function $d_G(i,j)$, we can express the average path length from $i$ to all other nodes as 
\begin{equation}
L_i = \frac{\sum_{j\in V\setminus i} d_G(i,j)}{|V|-1}.
\end{equation}
Hence, the closeness of a node can be formally written as 
\begin{equation}
C_c  (i) = \frac{1}{L_i} = \frac{|V|-1}{\sum_{j\in V\setminus i} d_G(i,j)}.
\end{equation}
The idea of shell number is based on a core to which a node lies in.  A $c$-core of a network is the subgraph in which every node is connected to the rest of the network by at least $c$ edges.  A node is assigned a shell number based on the maximal core that contains it.  This value can be derived exactly using shell decomposition~\cite{Seidman83}.  The eigenvector centrality~\cite{bona72} of a node is assigned based on the associated entry in the eigenvector of the adjacency matrix corresponding to the largest real eigenvalue.  The PageRank~\cite{Page98} for each node based on the PageRank of its neighbors.  An initial value for rank is considered for each node and the relationship is then computed iteratively until convergence is reached.  Intuitively, PageRank can be thought of as the importance of a node based on the importance of its neighbors.  We note that shell number, eigenvector, and PageRank are often associated with diffusion processes. A more complete discussion of centrality measures can be found in \cite{wasserman1994social}.

We evaluated the performance of centrality measures in finding a seed set by iteratively selecting the most central nodes with respect to a given measure until the $\Gamma_\theta$ of that set returns the set of all nodes.  Due to the computational overhead of calculating these centrality measures and the repeated re-evaluation of $\Gamma_\theta$, we limited this comparison to only \textbf{BlogCatalog3}, \textbf{CA-HepTh}, \textbf{CA-NetSci}, \textbf{URV E-Mail}, and \textbf{Douban} (no betweeness calcualted for \textbf{Douban}).  As with the experiments in the previous section, we studied threshold settings based on an integer in the interval $[1,10]$ (see Figure~\ref{centInt}) and as a fraction of incoming neighbors in the interval $[0.05,0.60]$ (multiples of $0.05$, see Figure~\ref{centFrac}).  In general, selecting highly-central nodes is an inefficient method for finding small seed sets.

In all but the lowest threshold settings, the use of centrality measures for the integer-threshold trials proved to significantly underperformed when the method presented in this paper - often returning seed-sets several orders of magnitude larger and in many cases including the majority of nodes in the network.  Even for the centrality measures outperformed our method in these trials, the reduction in seed set size was minimal (the greatest reduction in seed set size experienced in a centrality-measure test over the algorithm of this paper was $0.09\%$, while often producing seed sets orders of magnitude greater than our method).  This held even for the centrality measures associated with diffusion (shell number, eigenvector, and PageRank).

\begin{figure}
    \begin{center}
        \includegraphics[width=1\linewidth]{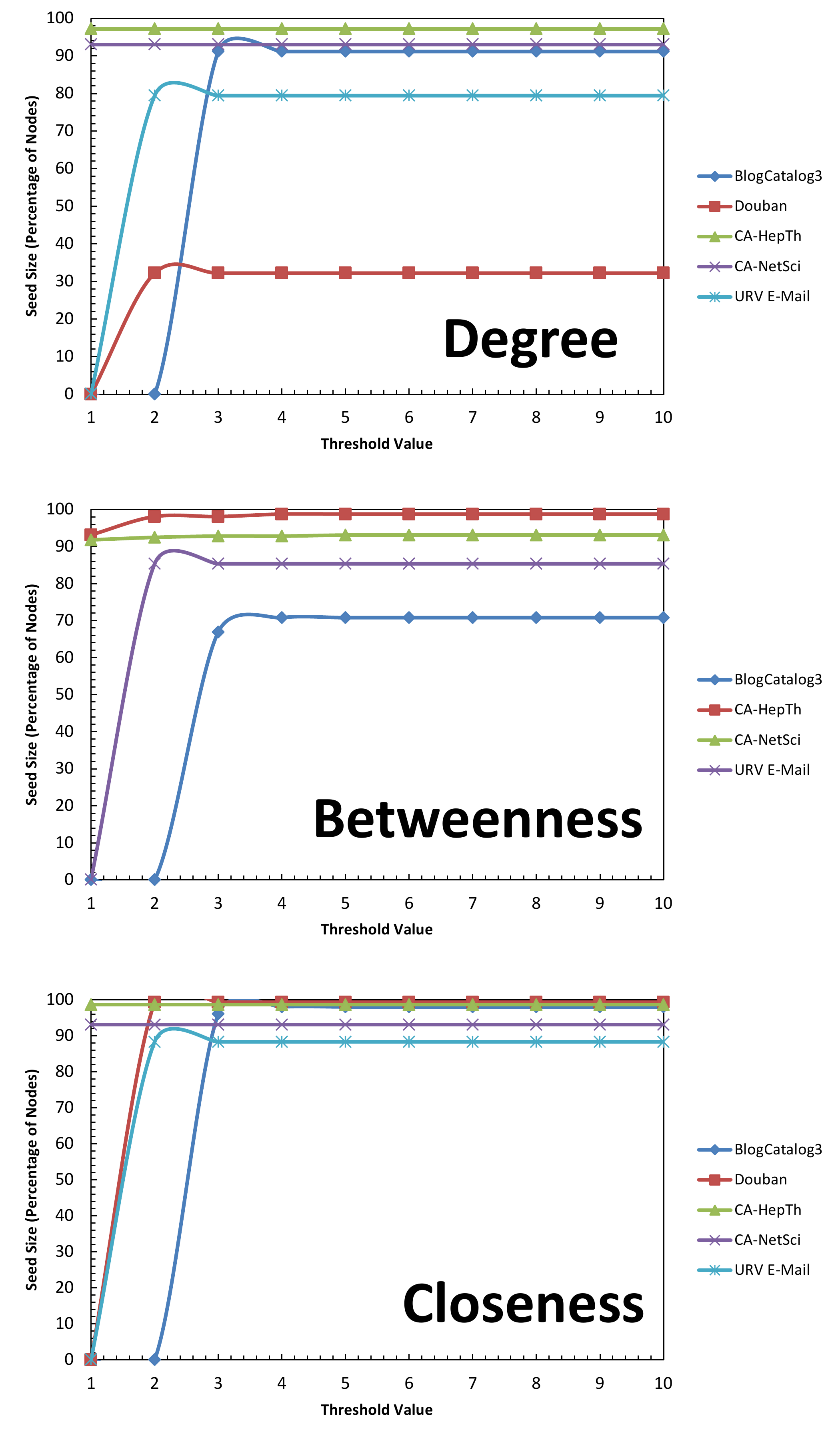}
    \end{center}
   \caption{The use of degree, betweenness and closeness to find seed-sets on select networks when the threshold is set to an integer in the interval $[1,10]$.  For these trials, centrality measure returned significantly larger (several orders of magnitude) larger seed sets than our approach.}
    \label{centInt}
\end{figure}

\begin{figure}
    \begin{center}
        \includegraphics[width=1\linewidth]{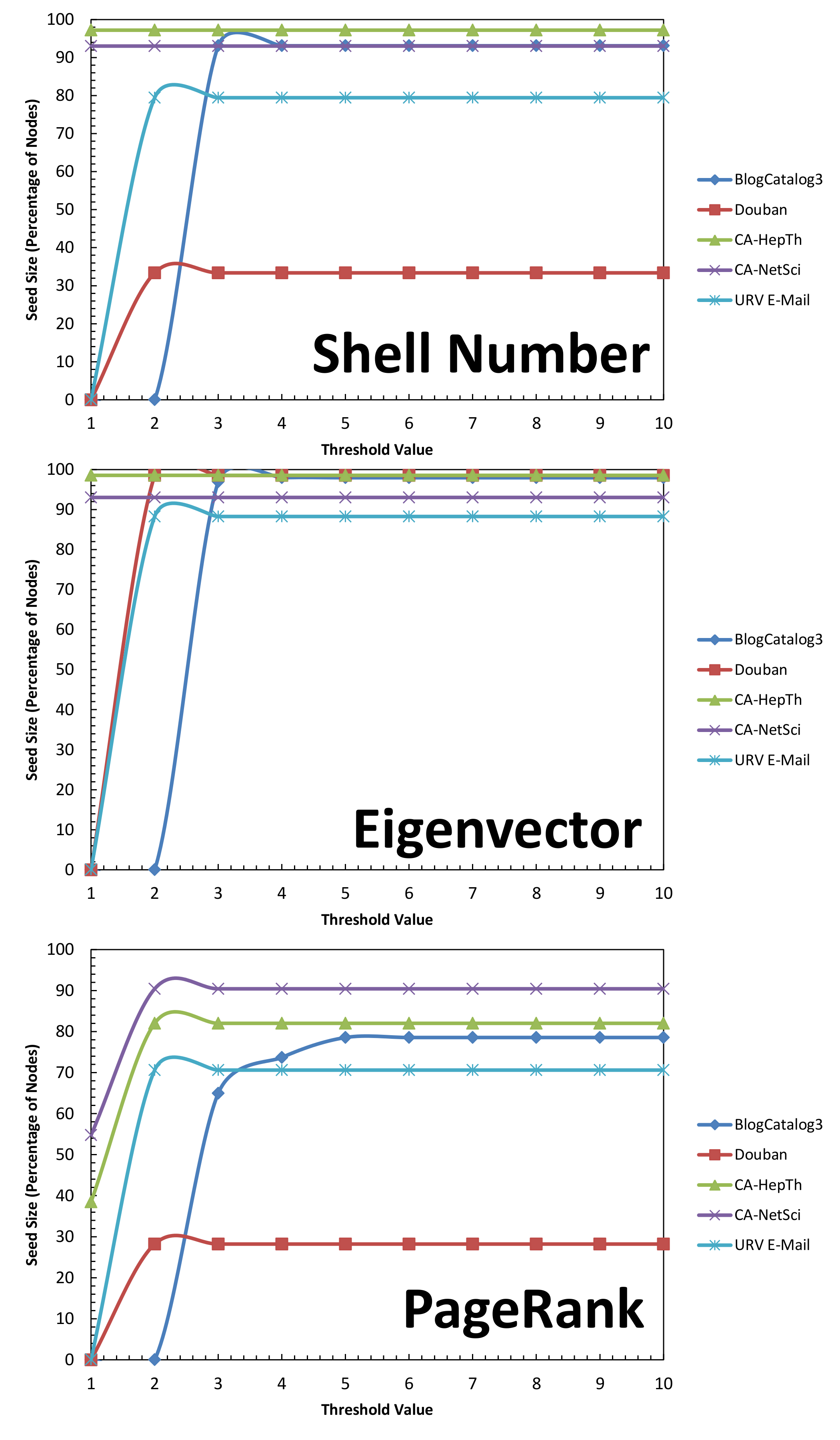}
    \end{center}
   \caption{The use of shell number, Eigenvector, and PageRank  to find seed-sets on select networks when the threshold is set to an integer in the interval $[1,10]$.}%  For these trials, centrality measure returned significantly larger (several orders of magnitude) larger seed sets than our approach.}
    \label{specInt}
\end{figure}

Our tests using fractional-based thresholds tell a slightly different story.  While our method still generally outperformed the centrality measures for the fractional tests, there were a few cases where the centrality measures fared better.  With \textbf{BlogCatalog3} all of the centrality measures outperformed our algorithm in the fraction-based experiments.  For that dataset, centrality-based algorithm consistently outperformed our method finding seed sets with less members (by $3.13-3.29\%$ of the population, on average).  With \textbf{URV-Email}, many trials that utilized a lower threshold setting outperformed our method, but never finding a feed set with smaller by more than $8\%$ of the total population.  However, in the larger threshold settings, our method consistently found smaller seeds.  For a given centrality measure for this dataset, centrality measures on average provided poorer results than our algorithm ranged - returning seed sets which were, on average $10.22-67.14\%$ (by overall population) larger than that returned by our algorithm.  Perhaps the most interesting result among the centrality measures were the PageRank fraction-based tests on \textbf{CA-NetSci}, which is associated with the largest seed sets.  PageRank found seed sets that were, on average $14.45\%$ smaller (by population) than our approach.  Additionally, though centrality measures outperformed \textsf{TIP\_DECOMP} for \textbf{BlogCatalog3}, this does not appear to hold for all social networks as the seed sets returned using centrality measures for the \textsf{Douban} approaches at least an order of magnitude increase over our method for nearly every fractional threshold setting for all centrality measures.  Hence, we conclude that for fraction-based thresholds, using centrality measures to find seed sets provides inconsistent results, and when it fails, it tends to provide a large portion of the network.  A possibility for a practical algorithm that could combine both methods would be to first run \textsf{TIP\_DECOMP}, returning some set $V'$.  Then, create set $V''$ by selecting the most central nodes until either $|V'| = |V''|$ or $\Gamma_\theta(V'')=V$ (whichever ensures the lower cardinality for $V''$.  If $|V'| = |V''|$, return $V'$, otherwise return $V''$.  For such an approach, we would likely recommend using degree centrality due to its ease of computation and performance in our experiments.  However, we note that highly-central nodes often may not be realistic targets for a viral-marketing campaign.  For instance, it may be cost-prohibitive to create a seed set consisting of major celebrities in order to spread the use of a product.  As such is a practical concern, we look at the performance of \textsf{TIP\_DECOMP} when high-degree nodes are removed in the next section.

\begin{figure}
    \begin{center}
        \includegraphics[width=1\linewidth]{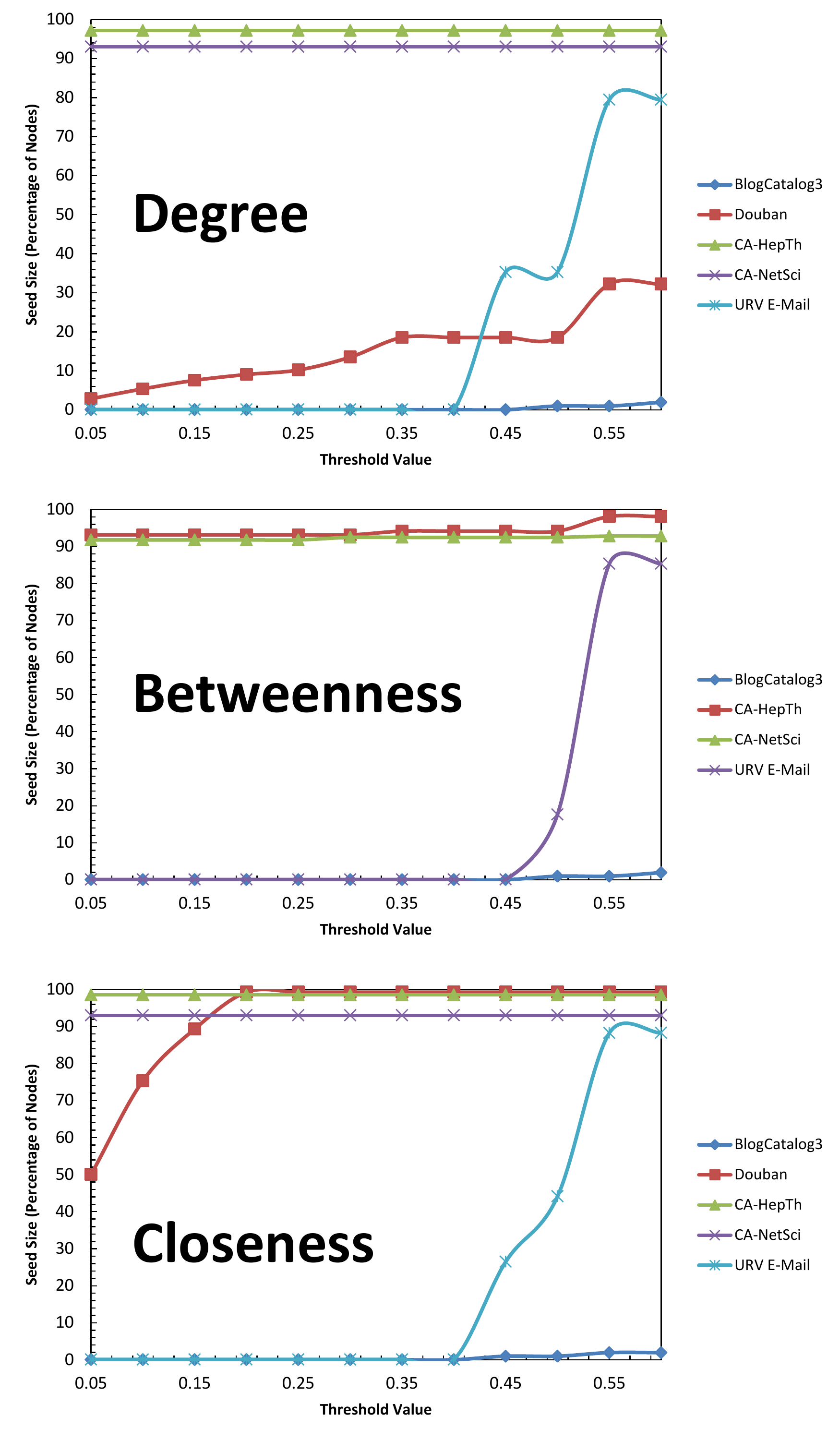}
    \end{center}
   \caption{The use of degree, betweenness and closeness to find seed-sets on select networks when the threshold is set to an fraction in the interval $[0.05,0.60]$ (multiples of $0.05$).}
    \label{centFrac}
\end{figure}

\begin{figure}
    \begin{center}
        \includegraphics[width=1\linewidth]{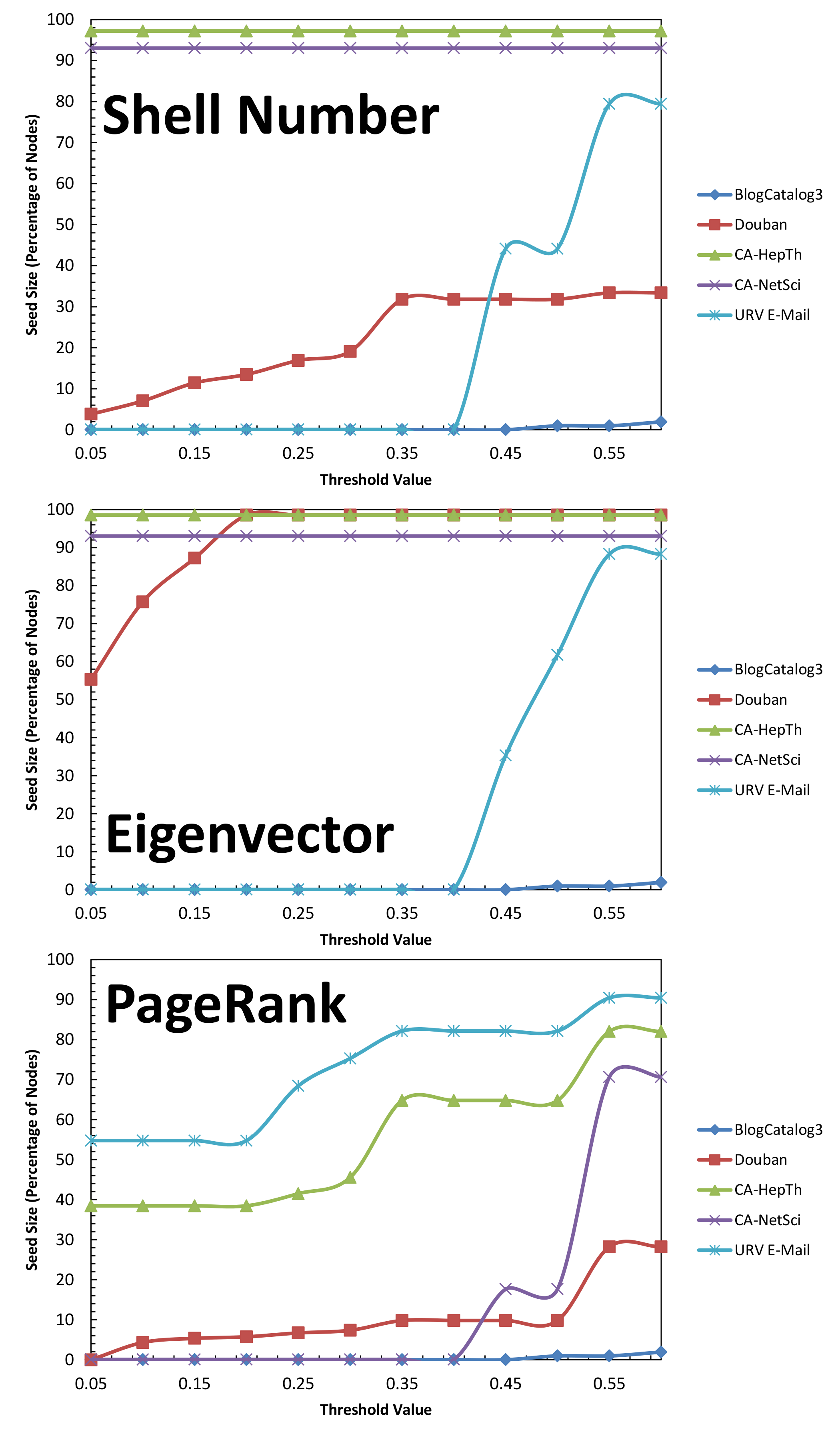}
    \end{center}
   \caption{The use of shell number, Eigenvector, and PageRank to find seed-sets on select networks when the threshold is set to an fraction in the interval $[0.05,0.60]$ (multiples of $0.05$).}
    \label{specFrac}
\end{figure}

%%%%%%%%%%%%%%%%%%%%%%%%%%%%%%%%%%%%%%%%%%%%%%%%%%%%%%%%

\subsection{The Speed of the Activation Process and Sets of ``Critical Mass''}
\label{speedSec}

An important aspect to consider in viral marketing is the speed of the activation process.  We illustrate this speed for several networks under a threshold of $2$ as well as a majority threshold (half of each nodes neighbors) in Figure~\ref{durTest}.  Interestingly, we found that the size of the initial seed set was not indicative of the speed of spreading.  For instance, in \textbf{BlogCatalog3}, a Category A network (for which our algorithm found a very small seed set) the activation process proceeded quickly when compared to the others examined.  However, this was also true for \textbf{CA-NetSci}, a Category C network (large seed set).  Conversely, the activation process in the \textbf{Douban} and \textbf{CA-HepTh} networks (also Category A and C, respectively) proceeded more slowly than the rest.

Another interesting feature we learned in exploring the speed of the activation process was that in all of our experiments there was a single time step where the number of activated nodes increased significantly more than the other time periods - sometimes by several orders of magnitude (see Figure~\ref{grtSingIncr}).  We can think of such a set of activated nodes as when the population reaches a ``critical mass'' which results in mass adoption in the next interval.  In many cases, such a critical mass is reached early - normally in the first few time-steps.

Finding a subset of the population of ``critical mass'' may be an important problem in its own right.  Though the critical mass point will often be larger than the seed set found by an algorithm in this paper, we can be assured that in one time step of the model,  the number of individuals reached (with a certain number of signals from their neighbors) is substantially larger than the investment.  In practice, this could lead to quicker spreading of information in an advertising campaign, for example.  Further, our experiments indicate that order-of-magnitude critical mass sets exist in several networks.  We are currently conducting further research on this topic.

\begin{figure}
    \begin{center}
        \includegraphics[width=1\linewidth]{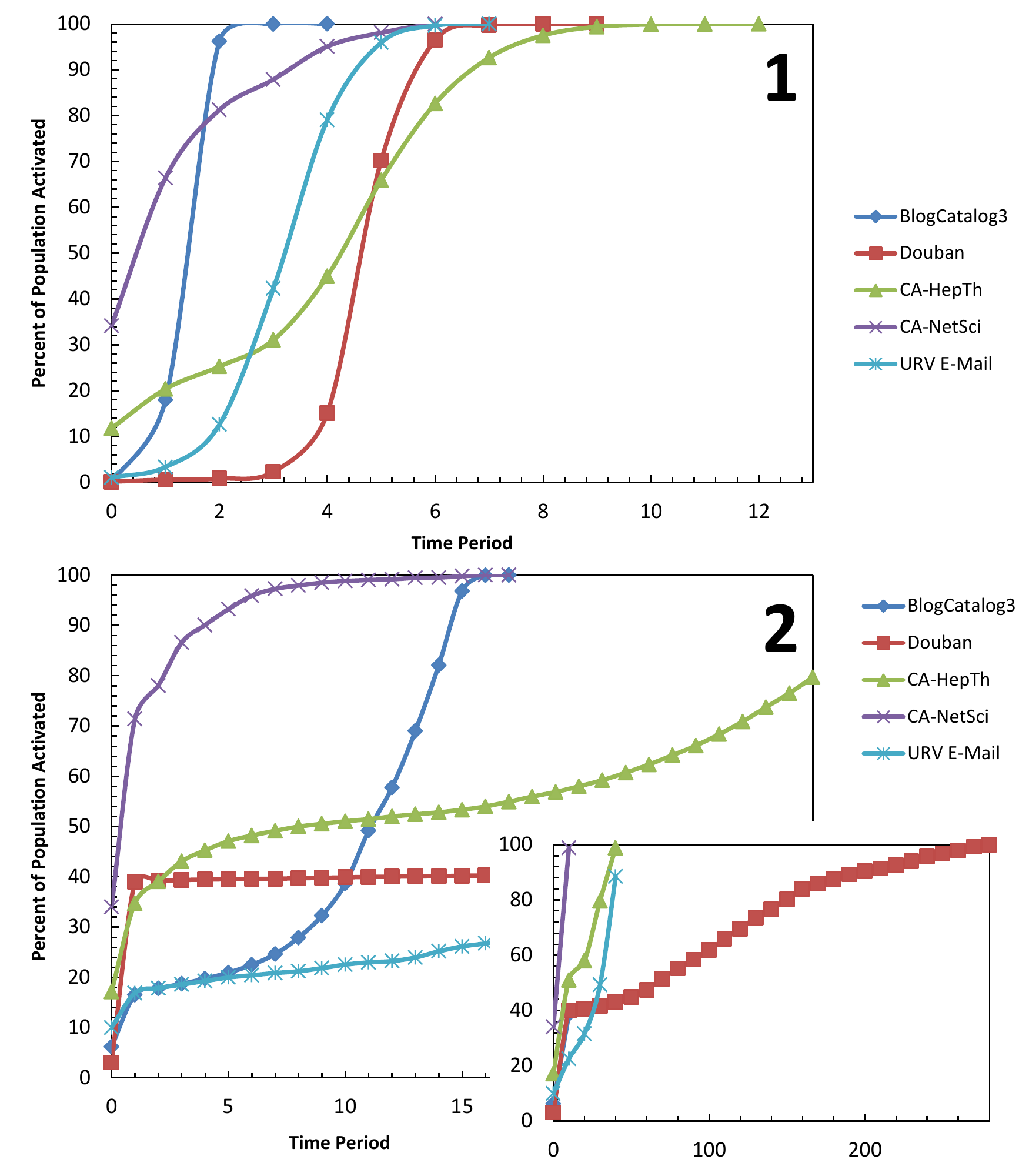}
    \end{center}
   \caption{An examination of several of speed of activation initiated from the seed set using a threshold of two (panel 1) and a majority threshold (panel 2).}
    \label{durTest}
\end{figure}

\begin{figure}
    \begin{center}
        \includegraphics[width=1\linewidth]{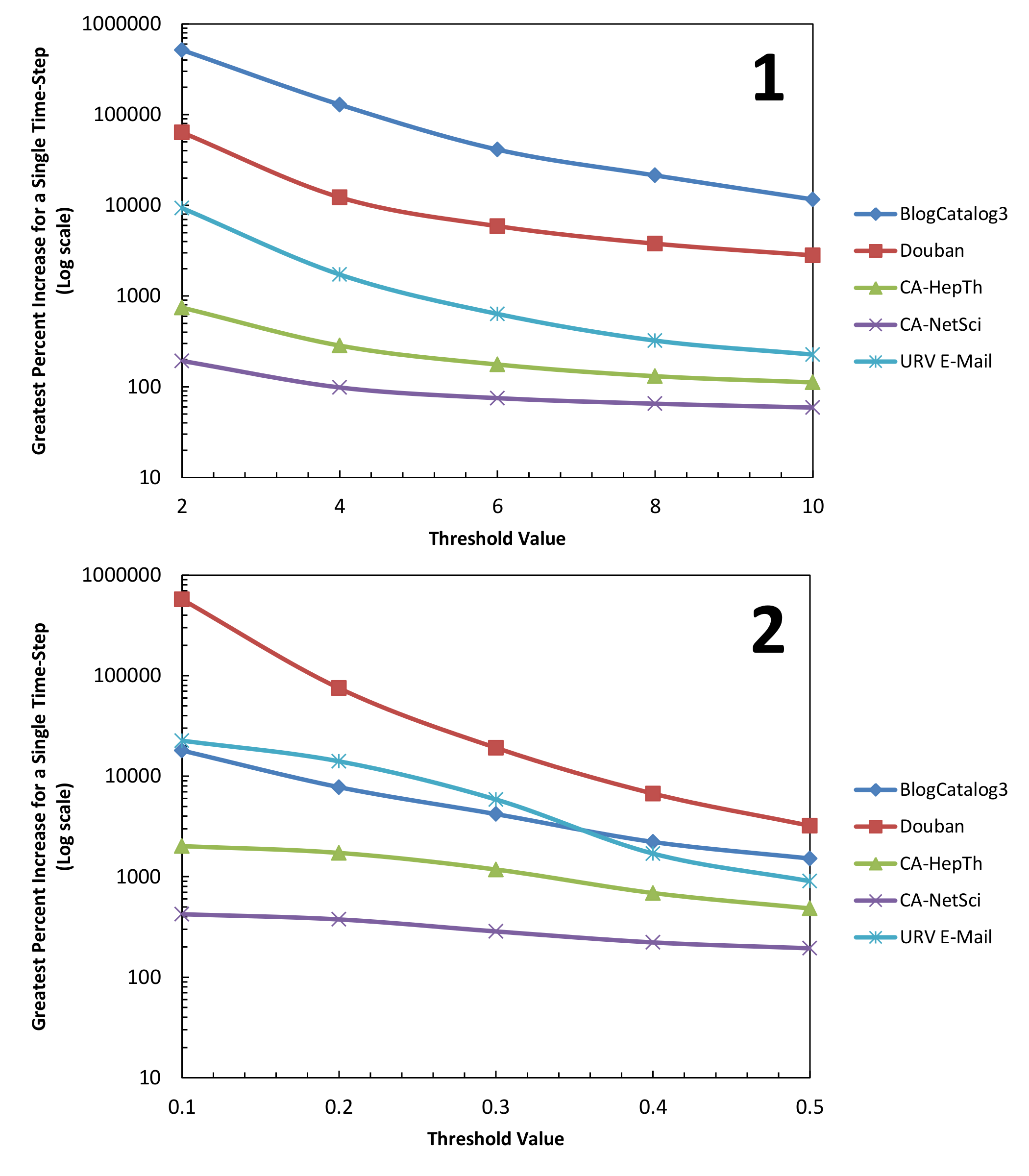}
    \end{center}
   \caption{Greatest Percent increase experienced in a single time step (the effect of reaching ``critical mass'') for integer-based and percentage-based thresholds (panel 1 and 2 respectively).}
    \label{grtSingIncr}
\end{figure}

%%%%%%%%%%%%%%%%%%%%%%%%%%%%%%%%%%%%%%%%%%%%%%%%%%%%%%%%

\subsection{Effect of Removing High-Degree Nodes}
\label{robust-sec}

In the last section we noted that high-degree nodes may not always be targetable in a viral marketing campaign (i.e. it may be cost prohibitive to include them in a seed set).  In this section, we explore the affect of removing high-degree nodes on the size of the seed-set returned by \textsf{TIP\_DECOMP}.  This type of node removal has also recently been studied in a different context in \cite{nodeRemRef}.  In these trials, we studied all $31$ networks and looked at two specific threshold settings: an integer threshold of $2$ (Figure~\ref{twoRemTest}) and a fractional threshold of $0.5$ (Figure~\ref{fifRemTest}).  We then studied the effect of removing up to $50\%$ of the nodes in order from greatest to least degree.

With an integer threshold of $2$, networks in category A still retained a seed-size (as returned by \textsf{TIP\_DECOMP}) two orders of magnitude smaller than the population size up to the removal of $10\%$ of the top degree nodes, and for many networks this was maintained to $50\%$.  Networks in category B retained seed sets an order of magnitude smaller than the population for up to $50\%$ of the nodes removed.  For most networks in category C, the seed size remained about a quarter of the population size up to $15\%$ of the top degree nodes being removed.

\begin{figure}
    \begin{center}
        \includegraphics[width=1\linewidth]{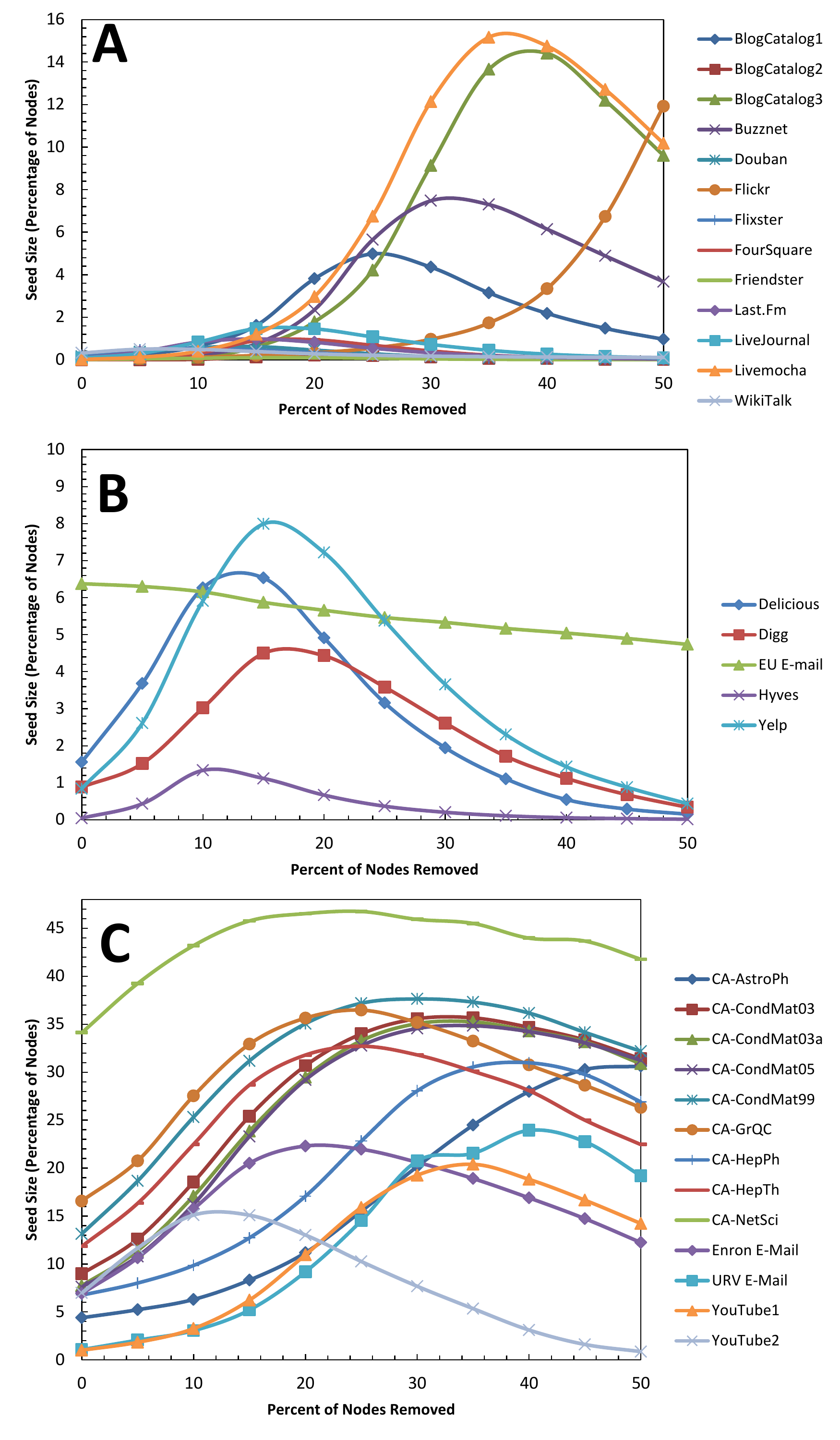}
    \end{center}
   \caption{Size of the seed set returned by \textsf{TIP\_DECOMP} (as a fraction of the popualtion) as a function of the percent of the highest degree nodes removed from the network with an integer theshold of $2$ for networks in categories A-C.}
    \label{twoRemTest}
\end{figure}

With a fractional threshold of $0.5$, we noted that many networks in category A actually had larger seed sets (as returned by \textsf{TIP\_DECOMP}) when more high degree nodes are removed.  Further, networks in categories A-B retained seed sets of at least an order of magnitude smaller than the population in these tests while most networks in category C retained sizes of about a quarter of the population.

\begin{figure}
    \begin{center}
        \includegraphics[width=1\linewidth]{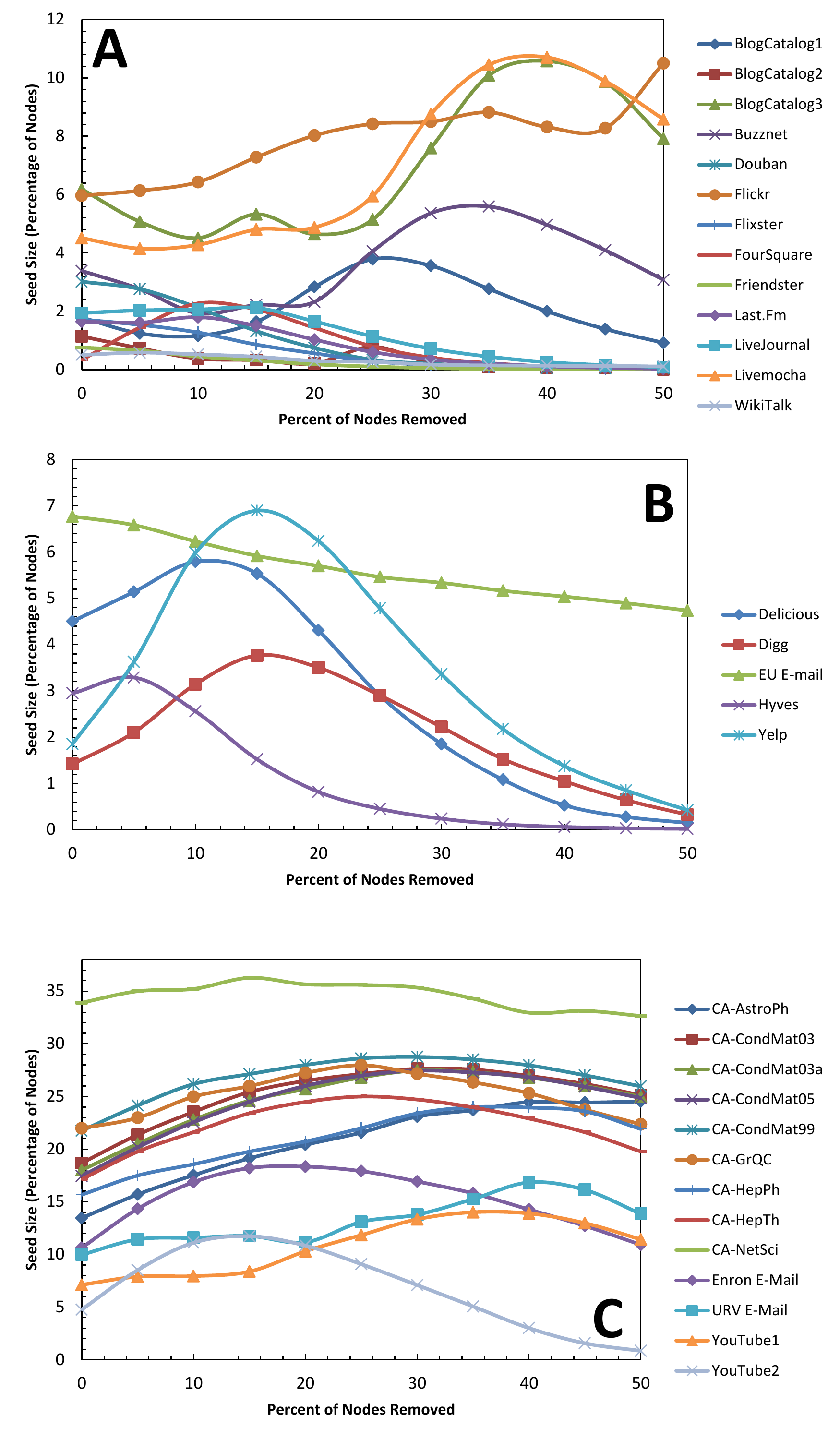}
    \end{center}
   \caption{Size of the seed set returned by \textsf{TIP\_DECOMP} (as a fraction of the popualtion) as a function of the percent of the highest degree nodes removed from the network with an fractional theshold of $0.5$ for networks in categories A-C.}
    \label{fifRemTest}
\end{figure}

\subsubsection{Seed Size as a Function of Community Structure}
\label{comm-sec}

In this section, we view the results of our heuristic algorithm as a measurement of how well a given network promotes spreading.  Here, we use this measurement to gain insight into which structural aspects make a  network more likely to be ``tipped.''  We compared our results with two network-wide measures characterizing community structure.  First, clustering coefficient ($C$) is defined for a node as the fraction of neighbor pairs that share an edge - making a triangle.  For the undirected case, we define this concept formally below.

\begin{definition}[Clustering Coefficient]
Let $ r $ be the number of edges between nodes with which $ v_i $ has an edge and $d_i$ be the degree of $v_i$.  The \textbf{clustering coefficient}, $ C_i = \dfrac {2r}  {d_i(d_i - 1)} $.
\end{definition}

Intuitively, a node with high $C_i$ tends to have more pairs of friends that are also mutual friends.  We use the average clustering coefficient as a network-wide measure of this local property.

Second, we consider modularity ($M$) defined by Newman and Girvan.~\cite{newman04}.  For a partition of a network, $M$ is a real number in $[-1,1]$ that measures the density of edges within partitions compared to the density of edges between partitions.  We present a formal definition for an undirected network below.

\begin{definition}[Modularity~\cite{newman04}]
Given partition $C= \{c_1,\ldots,c_q\}$, \textbf{ modularity}, 
\[
M = \dfrac 1 {2m} \sum_{c \in C}\sum_{i,j \in c}w_{ij}-P_{ij}
\]
where $m$ is the number of undirected edges; $w_{ij}=1$ if there is an edge between nodes $i$ and $j$ and $w_{ij}=0$ otherwise; $P_{ij}=\frac{k_i k_j}{2m}$
\end{definition}

%\begin{definition}[Modularity~\cite{newman04}]
%\textbf{Modularity}, $ M = \dfrac 1 {2m} \sum_{i,j \in V} [1 - \dfrac {d_i d_j} {2m}] \delta (c_i, c_j) $, where $m$ is the number of undirected edges, $d_i$ is node degree, $ c_i $ is the community to which $ v_i $ belongs and $ \delta (x, y) = 1 $ if $ x = y $ and $ 0 $ otherwise.
%\end{definition}

The modularity of an optimal network partition can be used to measure the quality of its community structure.  Though modularity-maximization is NP-hard, the approximation algorithm of Blondel et al.~\cite{blondel08} (a.k.a. the ``Louvain algorithm'') has been shown to produce near-optimal partitions.\footnote{Louvain modularity was computed using the implementation available from CRANS at  http://perso.crans.org/aynaud/communities/.}  We call the modularity associated with this algorithm the ``Louvain modularity.''  Unlike the $C$, which describes local properties, $M$ is descriptive of the community level.  For the $31$ networks we considered, $M$ and $C$  appear uncorrelated ($R^2 = 0.0538$, $p=0.2092$).

We plotted the initial seed set size ($S$) (from our algorithm - averaged over the $10$ threshold settings) as a function of $M$ and $C$ (Figure~\ref{main-fig}a) and uncovered a correlation (planar fit, $R^2=0.8666$, $p=5.666 \cdot 10^{-13}$, see Figure~\ref{main-fig} A).  The majority of networks in Category C (less susceptible to spreading) were characterized by relatively large $M$ and $C$ (Category C includes the top nine networks w.r.t. $C$ and top five w.r.t. $M$).  Hence, networks with dense, segregated, and close-knit communities (large $M$ and $C$) suppress spreading.  Likewise, those with low $M$ and $C$ tended to promote spreading.  Also, we note that there were networks that promoted spreading with dense and segregated communities, yet were less clustered (i.e. Category A networks Friendster and LiveJournal both have $M\geq 0.65$ and $C \leq 0.13$).  Further, some networks with a moderately large clustering coefficient were also in Category A (two networks extracted from BlogCatalog had $C\geq 0.46$) but had a relatively less dense community structure (for those two networks $M \leq 0.33$).

\begin{figure}
    \begin{center}
        \includegraphics[width=1\linewidth]{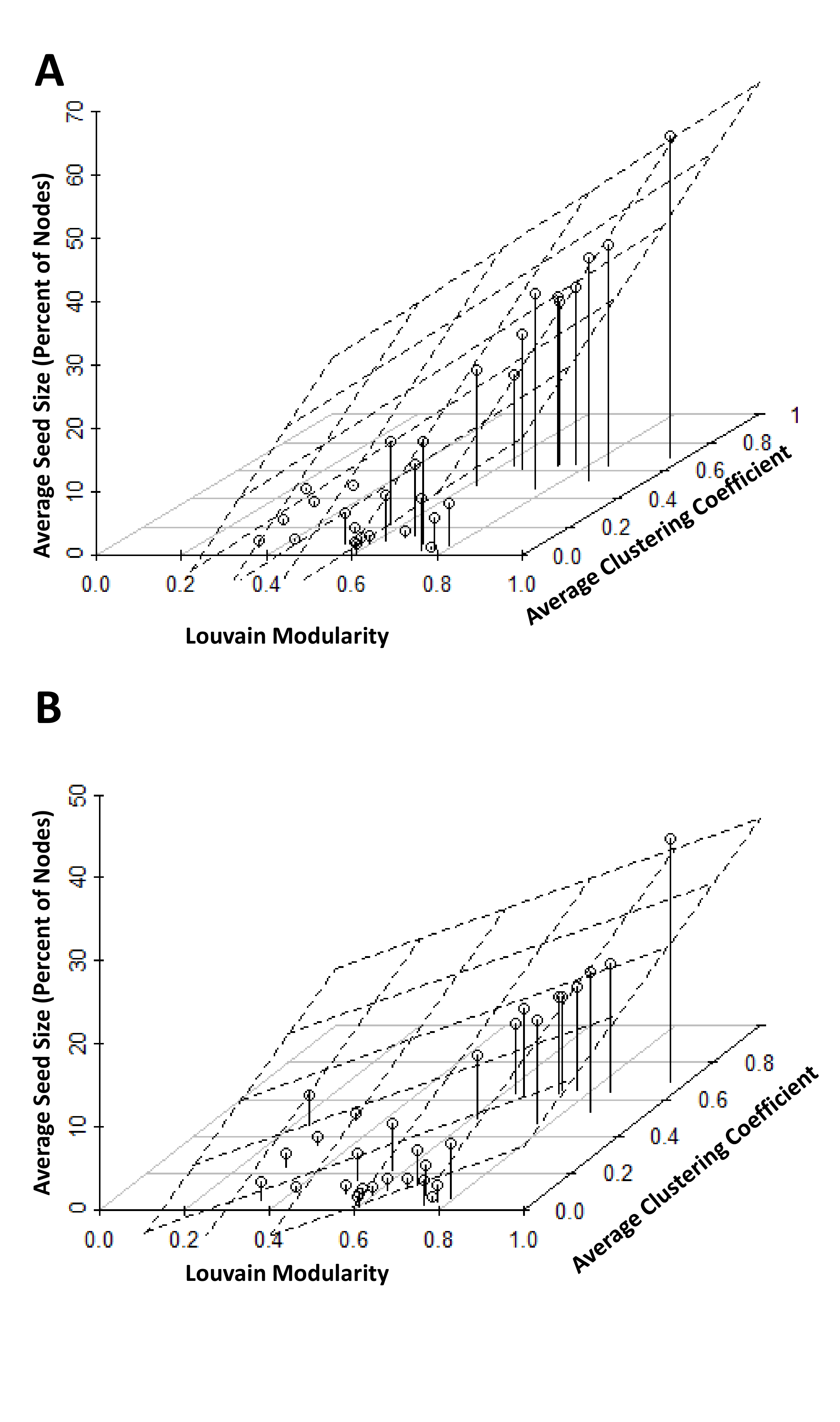}
    \end{center}
    \caption{\textbf{(A)} Louvain modularity ($M$) and average clustering coefficient ($C$) vs. the average seed size ($S$).  The planar fit depicted is $S=43.374 \cdot M +  33.794 \cdot C - 24.940$ with $R^2=0.8666$, $p=5.666 \cdot 10^{-13}$.  \textbf{(B)} Same plot at (A) except the averages are over the 12 percentage-based threshold values.  The planar fit depicted is $S=18.105 \cdot M + 17.257 \cdot C - 10.388$ with $R^2=0.816$, $p=5.117 \cdot 10^{-11}$.}
    \label{main-fig}
\end{figure}

\begin{table}[ht]
     \begin{center}
        \includegraphics[width=.8\linewidth]{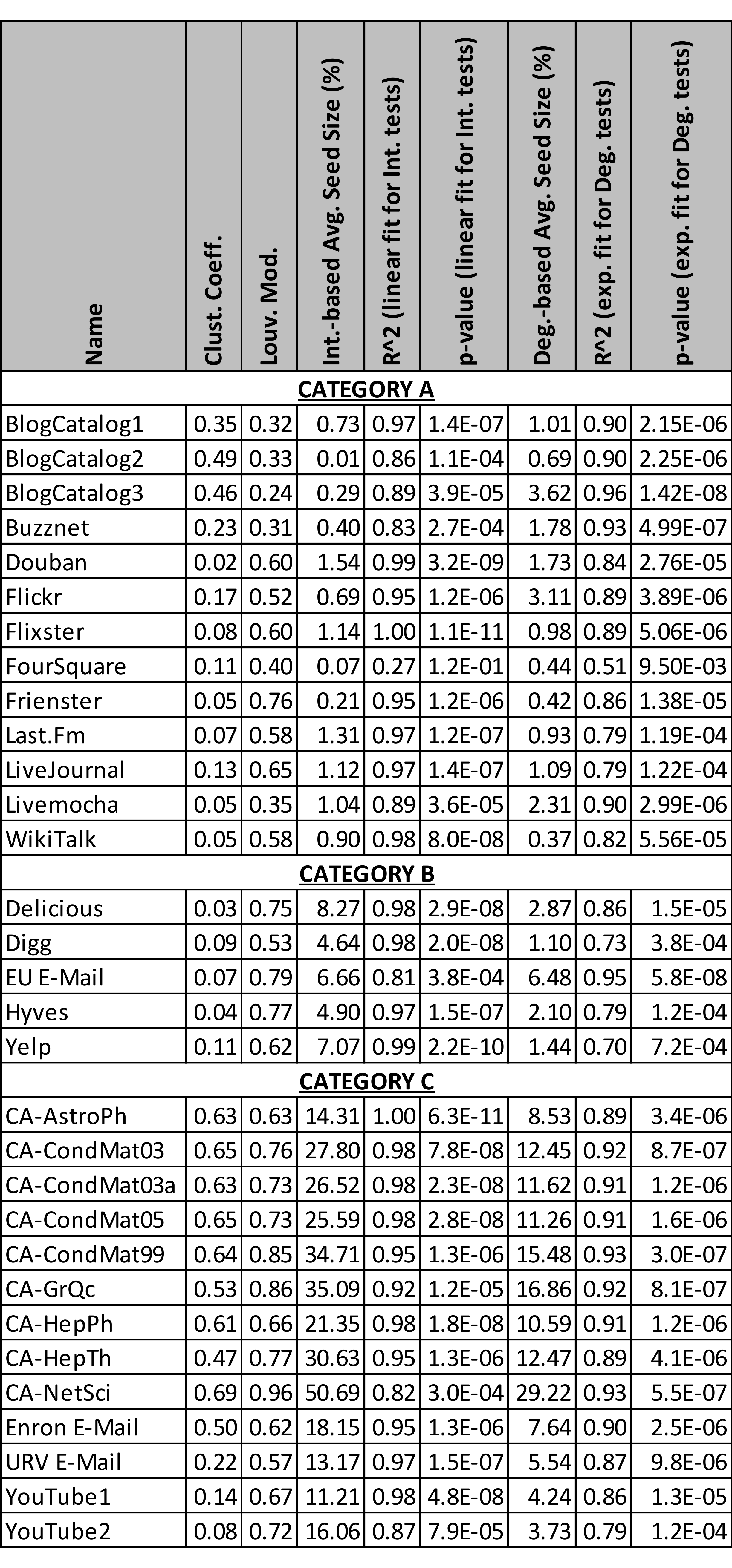}
    \end{center}
    \caption{Regression analysis and network-wide measures for the networks in Categories A, B, and C.}
    \label{figX}
\end{table}

\begin{table}[h]
     \begin{center}
        \includegraphics[width=.8\linewidth]{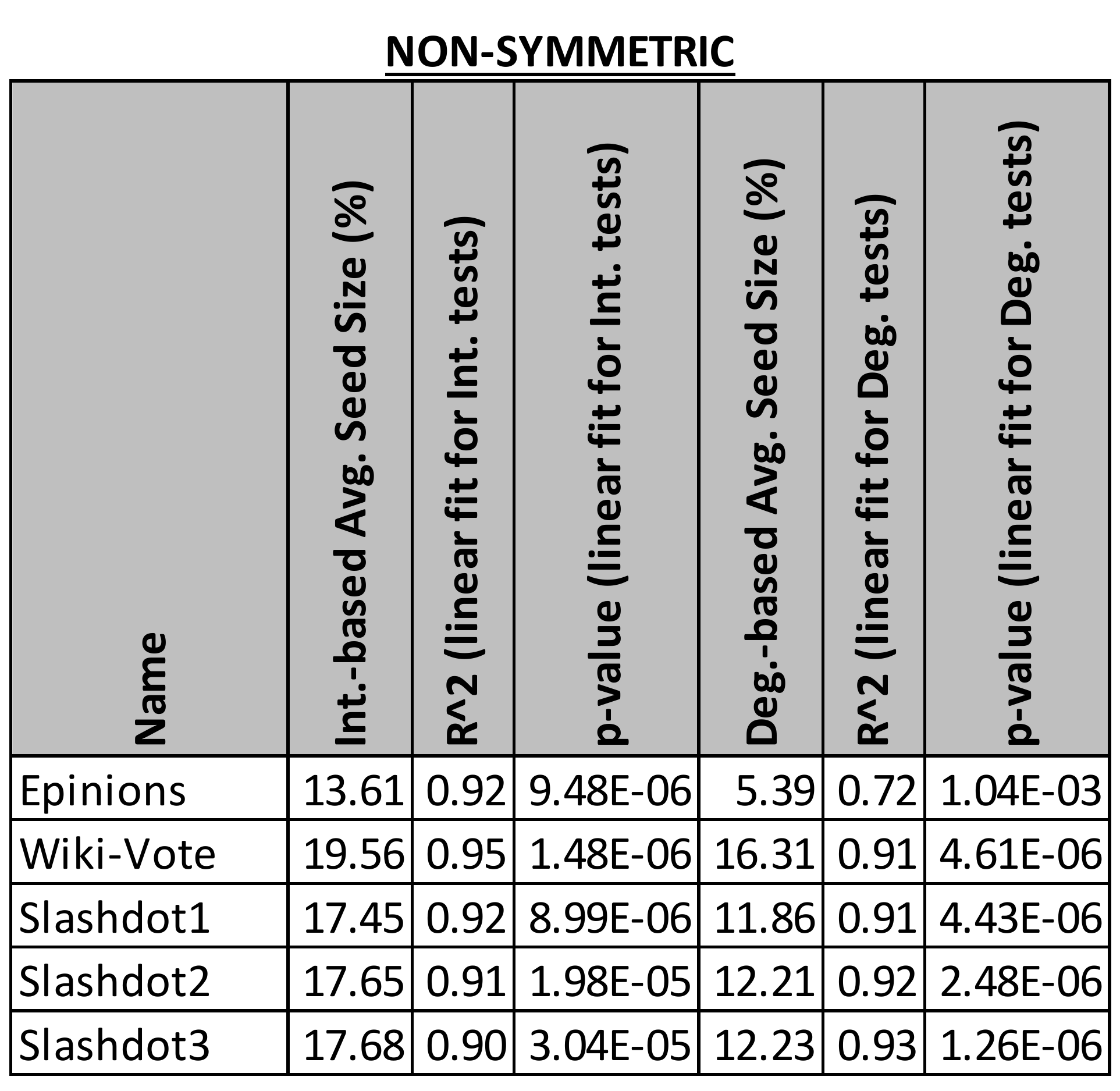}
    \end{center}
      \caption{Regression analysis and network-wide measures for the non-symmetric networks.}
    \label{dirResTable}
\end{table}

\section{Related Work}
\label{rw-sec}
Tipping models first became popular by the works of \cite{Gran78} and \cite{Schelling78} where it was presented primarily in a social context.  Since then, several variants have been introduced in the literature including the non-deterministic version of \cite{kleinberg} (described later in this section) and a generalized version of \cite{jy05}.  In this paper we focused on the deterministic version.  In \cite{wattsDodds07}, the authors look at deterministic tipping where each node is activated upon a percentage of neighbors being activated.  Dryer and Roberts \cite{Dreyer09} introduce the MIN-SEED problem, study its complexity, and describe several of its properties w.r.t. certain special cases of graphs/networks.  The hardness of approximation for this problem is described in \cite{chen09siam}.  The work of \cite{benzwi09} presents an algorithm for target-set selection whose complexity is determined by the tree-width of the graph - though it provides no experiments or evidence that the algorithm can scale for large datasets.  The recent work of \cite{reichman12} proves a non-trivial upper bound on the smallest seed set.

Our algorithm is based on the idea of shell-decomposition that currently is prevalent in physics literature.  In this process, which was introduced in \cite{Seidman83}, vertices (and their adjacent edges) are iteratively pruned from the network until a network ``core'' is produced.  In the most common case, for some value $k$, nodes whose degree is less than $k$ are pruned (in order of degree) until no more nodes can be removed.  This process was used to model the Internet in \cite{ShaiCarmi07032007} and find key spreaders under the SIR epidemic model in \cite{InfluentialSpreaders_2010}.  More recently, a ``heterogeneous'' version of decomposition was introduced in \cite{baxter11} - in which each node is pruned according to a certain parameter - and the process is studied in that work based on a probability distribution of nodes with certain values for this parameter.

\subsection{Notes on Non-Deterministic Tipping}
We also note that an alternate version of the model where the thresholds are assigned randomly has inspired approximation schemes for the corresponding version of the seed set problem.\cite{kleinberg,leskovec07,chen10}  Work in this area focused on finding a seed set of a certain size that maximizes the expected number of adopters.  The main finding by Kempe et al., the classic work for this model, was to prove that the expected number of adopters was submodular - which allowed for a greedy approximation scheme.  In this algorithm, at each iteration, the node which allows for the greatest increase in the expected number of adopters is selected.  The approximation guarantee obtained (less than $0.63$ of optimal) is contingent upon an approximation guarantee for determining the expected number of adopters - which was later proved to be $\#P$-hard.~\cite{chen10}  Recently, some progress has been made toward finding a guarantee~\cite{Borgs12}.  Further, the simulation operation is often expensive - causing the overall time complexity to be $O(x \cdot n^2)$ where $x$ is the number of runs per simulation and $n$ is the number of nodes (typically, $x>n$).  In order to avoid simulation, various heuristics have been proposed, but these typically rely on the computation of geodesics - an $O(n^3)$ operation - which is also more expensive than our approach.

Additionally, the approximation argument for the non-deterministic case does not directly apply to the original (deterministic) model presented in this paper.  A simple counter-example shows that sub-modularity does not hold here. Sub-modularity (diminishing returns) is the property leveraged by Kempe et al. in their approximation result.

\subsection{Note on an Upper Bound of the Initial Seed Set}

Very recently, we were made aware of research by Daniel Reichman that proves an upper bound on the minimal size of a seed set for the special case of undirected networks with homogeneous threshold values.~\cite{reichman12}  The proof is constructive and yields an algorithm that mirrors our approach (although Reicshman's algorithm applies only to that special case).  We note that our work and the work of Reichman were developed independently.  We also note that Reichman performs no experimental evaluation of the algorithm.

Given undirected network $G$ where each node $v_i$ has degree $d_i$ and the threshold value for all nodes is $k$, Reichman proves that the size of the minimal seed set can be bounded by $\sum_i \min\{1, \frac{k}{d_i+1}\}$.  For our integer tests, we compared our results to Reichman's bound.  Our seed sets were considerably smaller - often by an order of magnitude or more.  See Figure~\ref{fig4a} for details.

\section{Conclusion}

As recent empirical work on tipping indicates that it can occur in real social networks,\cite{centola10,zhang11} our results are encouraging for viral marketers.  Even if we assume relatively large threshold values, small initial seed sizes can often be found using our fast algorithm - even for large datasets.  For example, with the FourSquare online social network, under majority threshold ($50\%$ of incoming neighbors previously adopted), a viral marketeer could expect a $297$-fold return on investment.  As results of this type seem to hold for many online social networks, our algorithm seems to hold promise for those wishing to ``go viral.''  An important open question to address in future work is if a similar decomposition-based approach is viable for finding seed sets under other diffusion models, such as the independent cascade model~\cite{kleinberg} and evolutionary graph theory~\cite{lieberman_evolutionary_2005} as well as probabilistic variants of the tipping model and diffusion processes on multi-modal networks~\cite{snops-iclp}.  Exploring other models can lead to the development of decomposition algorithms in domains where social behavior is more dynamic such as cell-phone networks~\cite{dyag13,otherPho}.

\begin{acknowledgements}
We would like to thank Gaylen Wong (USMA) for his technical support.  Additionally, we would like to thank (in no particular order) Albert-L\'{a}szl\'{o} Barab\'{a}si (NEU), Sameet Sreenivasan (RPI), Boleslaw Szymanski (RPI), Patrick Roos (UMD), John James (USMA), and Chris Arney (USMA) for their discussions relating to this work.  Finally, we would also like to thank Megan Kearl, Javier Ivan Parra, and Reza Zafarani of ASU for their help with some of the datasets.  
The authors are supported under by the Army Research Office (project 2GDATXR042) and the Office of the Secretary of Defense (project F1AF262025G001).  The opinions in this paper are those of the authors and do not necessarily reflect the opinions of the funders, the U.S. Military Academy, or the U.S. Army.
\end{acknowledgements}

\begin{figure}
    \begin{center}
        \includegraphics[width=1\linewidth]{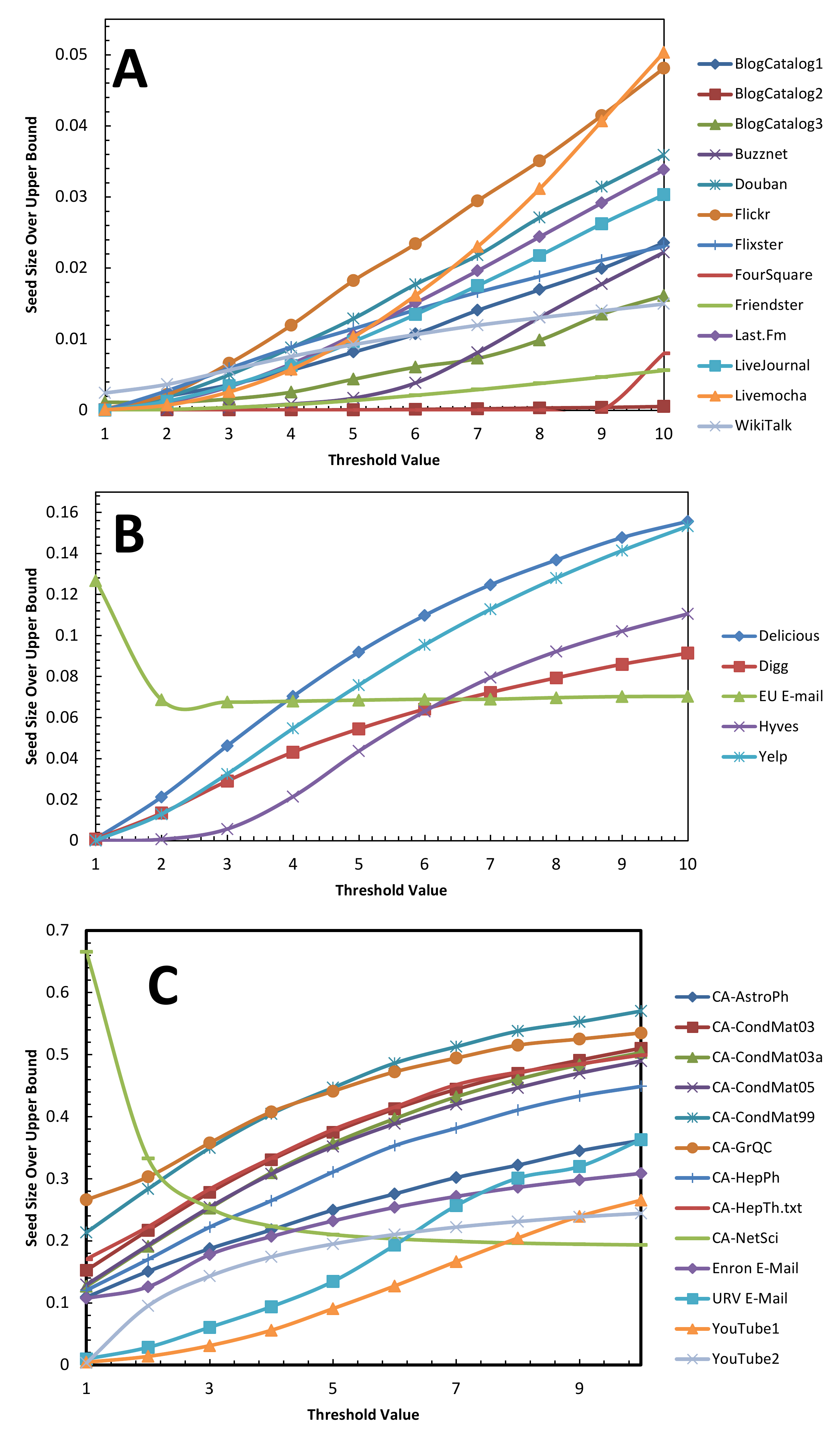}
    \end{center}
    \caption{Integer threshold values vs. the seed size divided by Reichman's upper bound~\cite{reichman12} the three categories of networks (categories A-C are depicted in panels A-C respectively).  Note that in nearly every trial, our algorithm produced an initial seed set significantly smaller than the bound - in many cases by an order of magnitude or more.}
    \label{fig4a}
\end{figure}

%\bibliographystyle{spmpsci}      % mathematics and physical sciences
%\bibliography{network}   % name your BibTeX data base

\end{document}